\documentclass[preprint,hidelinks,onefignum,onetabnum]{siamart220329}

\usepackage[utf8]{inputenc} 
\usepackage[T1]{fontenc}    
\usepackage{hyperref}       
\usepackage{url}            
\usepackage{booktabs}       
\usepackage{amsfonts}       
\usepackage{amsmath}
\usepackage{nicefrac}       
\usepackage{microtype}      

\usepackage{sgame}
\usepackage{upgreek}
\usepackage{bm}
\usepackage{amsmath,amssymb,amsfonts}

\usepackage{caption}
\usepackage{subcaption}

\usepackage[algo2e, linesnumbered, ruled,noend]{algorithm2e}		

\newcommand{\nn}{\mathbb{N}}
\newcommand{\rr}{\mathbb{R}}
\renewcommand{\tt}{\mathbb{T}}
\newcommand{\uu}{\mathbb{U}}
\newcommand{\xx}{\mathbb{X}}
\newcommand{\yy}{\mathbb{Y}}
\newcommand{\zz}{\mathbb{Z}}

\newcommand{\NN}{\mathcal{N}}
\newcommand{\PP}{\mathcal{P}}
\newcommand{\GG}{\mathcal{G}}

\renewcommand{\P}{{\rm Pr}}

\newcommand{\ba}{\textbf{a}}
\newcommand{\bd}{\textbf{d}}
\newcommand{\bU}{\textbf{U}}
\newcommand{\bu}{\textbf{u}}

\newcommand{\bT}{\textbf{T}}

\newcommand{\bpi}{\bm{\pi}}
\newcommand{\bPi}{\bm{\Pi}}

\newcommand{\bGamma}{\bm{\Gamma}}

\newcommand{\DS}{{\rm DS}}
\newcommand{\sat}{{\rm Sat}}

\newcommand{\uniform}{\text{Unif}}

\newcommand{\BR}{{\rm BR}}

\newcommand{\eq}[1][0]{\bm{\Gamma}^{#1{\rm \text{-}eq}}   }

\newtheorem{assumption}{Assumption}

\newcommand{\by}[1]{{\color{black}#1}}

\begin{document}

\title{Satisficing Paths and Independent Multi-Agent Reinforcement Learning in Stochastic Games\thanks{This work was supported in part by the Natural Sciences and Engineering Research Council of Canada.}
}

\author{Bora Yongacoglu\thanks{Department of Mathematics and Statistics, Queen's University}
\and G\"{u}rdal Arslan\thanks{Department of Electrical Engineering, University of Hawaii at Manoa}
\and Serdar Y\"{u}ksel\footnotemark[2]
}

\maketitle

\begin{abstract}%
In multi-agent reinforcement learning (MARL), independent learners are those that do not observe the actions of other agents in the system. Due to the decentralization of information, it is challenging to design independent learners that drive play to equilibrium. This paper investigates the feasibility of using \emph{satisficing} dynamics to guide independent learners to approximate equilibrium in stochastic games. For $\epsilon \geq 0$, an $\epsilon$-satisficing policy update rule is any rule that instructs the agent to not change its policy when it is $\epsilon$-best-responding to the policies of the remaining players; $\epsilon$-satisficing paths are defined to be sequences of joint policies obtained when each agent uses some $\epsilon$-satisficing policy update rule to select its next policy. We establish structural results on the existence of $\epsilon$-satisficing paths into $\epsilon$-equilibrium in both symmetric $N$-player games and general stochastic games with two players. We then present an independent learning algorithm for $N$-player symmetric games and give high probability guarantees of convergence to $\epsilon$-equilibrium under self-play. This guarantee is made using symmetry alone, leveraging the previously unexploited structure of $\epsilon$-satisficing paths. 
\end{abstract}



\author{Bora Yongacoglu \thanks{Department of Mathematics and Statistics, Queen's University}
\and G\"{u}rdal Arslan \thanks{Department of Electrical Engineering, University of Hawaii at Manoa}
\and Serdar Y\"{u}ksel\footnotemark[1]
}

\section{Introduction}    
Reinforcement learning (RL) algorithms use experience and feedback information to improve one's performance in a control task \cite{sutton2018reinforcement}. In recent years, the field of RL has advanced tremendously both in terms of fundamental theoretical contributions (e.g. \cite{agarwal2021theory}, \by{\cite{schulman2015trust,schulman2017proximal}}) and successful applications (e.g. \cite{silver2016mastering,silver2017mastering}, \cite{mnih2015human},\cite{brown2018superhuman}). These advances have led to the deployment of RL algorithms in large-scale systems in which many agents act, observe, and learn in a shared environment. Multi-agent reinforcement learning (MARL) is the study of emergent behaviour in complex, strategic environments, and is one of the important frontiers in modern artificial intelligence research.

The literature on MARL is relatively small when compared to that of single-agent RL, and this owes largely to the inherent challenges of learning in multi-agent settings. The first such challenge is of decentralized information: some relevant information will be unavailable to some of the players. This may occur due to strategic considerations, as competing agents may wish to hide their actions or knowledge from their rivals (as studied in \cite{ornik2018deception}), or it may occur simply because of obstacles in communicating, observing, or storing large quantities of information in decentralized systems. 

The second challenge inherent to MARL comes from the non-stationarity of the environment from the point of view of any individual agent (see, for instance, the survey by \cite{hernandez2017survey}). As an agent learns how to improve its performance, it will alter its behaviour, and this can have a destabilizing effect on the learning processes of the remaining agents, who may change their policies in response to outdated strategies. Notably, this issue arises when one tries to apply single-agent RL algorithms---which typically rely on state-action value estimates or gradient estimates that are made using historical data---in multi-agent settings. A number of studies, including \cite{tan1993multi} and \cite{Claus1998}, have reported non-convergent play when single-agent algorithms using local information are employed, without modification, in multi-agent settings.  

Designing decentralized learning algorithms with desirable convergence properties is a task of great practical importance that lies at the intersection of the two challenges above. The notion of decentralization considered in this paper involves agents that observe a global state variable but do not observe the actions of other agents. Learning algorithms suitable for this information structure are called \textit{independent learners} in the machine learning literature \cite{zhang2021multi, matignon2012survey, matignon2009coordination, wei2016lenient}; they have also been called \textit{payoff-based} and \textit{radically uncoupled} in the control and game theory literatures, respectively, \cite{Marden2009payoff, marden2012revisiting, foster2006regret}.

For our theoretical framework, we consider stochastic games with discounted costs. In this setting, our overarching goal is to provide MARL algorithms that are suitable for independent learners in a complex system, require little coordination among agents, and come with provable guarantees for long-run performance. To inform the development of such algorithms, this paper identifies structural properties of games that can be leveraged in algorithm design. We then illustrate the usefulness of the identified structure by providing an independent learning algorithm and proving that, under mild conditions, this algorithm leads to approximate equilibrium policies in self-play. 

The structure we consider relates to \textit{satisficing}, a natural approach to optimization that, as we discuss in \S\ref{ss:literature} and \S\ref{sec:epsilon-paths}, is used in several existing independent MARL algorithms. An agent that uses satisficing searches its policy space until it finds a policy that is deemed satisfactory and sufficient, at which point it settles on this policy. The agent continues to use this policy as long as the policy remains satisfactory. At a high level, the satisficing paths property formalized in \S\ref{sec:epsilon-paths} holds for a game and a subset of joint policies if there exist policy update rules of the satisficing variety that can drive play to equilibrium from any initial policy in the given policy subset. We show that two important classes of games---namely symmetric $N$-player games and general two-player games---admit this property within the set of stationary policies, which suggests that independent MARL algorithms that employ satisficing to update policies can be used to drive play to equilibrium in such games. 
 
For $N$-player symmetric stochastic games, we build on this finding to present an algorithm that drives play to approximate equilibrium. This algorithm uses the exploration phase technique of \cite{AY2017} for policy evaluation, but differs in how players update their policies. Here, players discretize their policy space with a quantizer and use a satisficing rule to explore this quantized set, occasionally using random search when unsatisfied. Of note, here we do not restrict players to using deterministic stationary policies (pure strategies), as was done in \cite{AY2017}, and allow for use of randomized stationary policies (mixed strategies), enabling convergence to near equilibrium in games that do not admit near equilibria in the set of deterministic stationary policies.

By relying on the satisficing paths property formalized in \S\ref{sec:epsilon-paths}, our proof of convergence does not assume any further structure in the game beyond symmetry. To our knowledge, this is the first algorithm with formal convergence guarantees in this class of games: as we will discuss below, previous rigorous work on independent learners has focused on different---highly structured---classes of games, such as teams, potential games, weakly acyclic games, and two-player zero-sum games.

\vspace{5.5pt} 

\noindent \textbf{Contributions:}
\begin{itemize}
	\item[(i)] For any stochastic game and $\epsilon \geq 0$, we define \textit{$\epsilon$-satisficing paths} (Definition~\ref{def:epsilon-path}) and a related  \textit{$\epsilon$-satisficing paths property} (Definition~\ref{def:paths-property});
	
	\item[(ii)] In Theorem \ref{thm:symmetric-paths}, we prove that  symmetric games have the $\epsilon$-satisficing paths property, for all $\epsilon \geq 0$. \by{Moreover, our proof technique shows that, in symmetric games, the $\epsilon$-satisficing paths property is compatible with quantization, provided the quantization is sufficiently fine and symmetric;}
	
	\item[(iii)] In Theorem~\ref{thm:two-player-paths}, we prove that any two-player game has the $\epsilon$-satisficing paths property, for all $\epsilon$ \by{$\geq$} $0$;
	
	\item[(iv)] We present  Algorithm~\ref{algo:main} for symmetric stochastic games and, in Theorem \ref{thm:main}, we prove that self-play drives the policy process to $\epsilon$-equilibrium.

\end{itemize}

\subsection{Related Work} \label{ss:literature}

Beginning with Brown's fictitious play algorithm \cite{Brown1951iterative,Robinson1951iterative}, the study of learning in games is nearly as old as game theory itself. There is a large \by{and ongoing} literature on fictitious play and its variants, with most works in this line considering a different information structure than the decentralized one studied here. The bulk of work on fictitious play focuses on settings with perfect monitoring of the actions of other players. \by{This tradition includes the recent works of \cite{leslie2020best} and \cite{sayin2022fictitious}, which study fictitious play-type algorithms with perfect monitoring in stochastic games. Additionally, multiple recent studies have considered fictitious play-type algorithms for various decentralized information structures, such as \cite{swenson2018distributed, eksin2017distributed}, and \cite{sayin2021decentralized}.}

A number of early empirical works studied the behaviour resulting from independent RL agents coexisting in various shared environments, e.g. \cite{tan1993multi, sen94, Claus1998}. Contemporaneously, stochastic games were proposed as a theoretical framework for MARL  \cite{Littman1994}. Several \textit{joint action learners} (learners that require access to the past actions of all other agents) were then proposed for playing stochastic games and proven to converge to equilibrium under various assumptions. A representative sampling of this stream of algorithms includes the Minimax Q-learning algorithm of \cite{Littman1994}, the Nash Q-learning algorithm of \cite{Hu2003}, and the Friend-or-Foe Q-learning algorithm of \cite{Littman2001ffq}. 

Early work on independent learners includes the following: \cite{Claus1998} popularized the terminology of joint action learners and independent learners and stated conjectures; \cite{lauer2000} presented an independent learner for fully cooperative games with deterministic state transitions and cost realizations and proved its convergence to optimality in that setting; and \cite{Bowling} proposed the WoLF-Policy Hill Climbing algorithm for general-sum stochastic games and conducted simulation studies.

Due in part to the challenges posed by non-stationarity and decentralized information, most contributions to the literature on independent learners focused either on the stateless case of repeated games and produced formal results, such as the works of \cite{leslie2005individual, foster2006regret, germano2007global, Chasparis2013aspiration, Marden2009payoff, marden2012revisiting,  marden2014achieving}, or otherwise studied the multi-state setting and presented only empirical results, such as the works \cite{matignon2007hysteretic, matignon2009coordination, wei2016lenient}.

More recently, a number of papers have studied independent learners for games with non-trivial state dynamics while still presenting rigorous guarantees. In \cite{daskalakis2021independent}, the authors studied the convergence of single-agent policy gradient algorithms employed in episodic two-player zero-sum games. It was shown that if the players' policy updates satisfy a particular two-timescale rule, with one player updating quickly and the other updating slowly, then policies approach an approximate equilibrium. A complementary study was conducted in \cite{sayin2021decentralized}, where a different learning rule was proposed for non-episodic two-player zero-sum games. In this setting, a convergence result for the value function estimates was provided.

The preceding works produce rigorous results by taking advantage of the considerable structure of two-player zero-sum games, which are inherently adversarial strategic environments. Another class of games possessing very different exploitable structure is that of stochastic teams and their generalizations of weakly acyclic games and common interest games. An independent learning algorithm for weakly acyclic games was presented in \cite{AY2017}. By synchronized policy updating, this algorithm is able to drive play to equilibrium via inertial best-response dynamics. In a recent paper \cite{YAY-TAC}, we modify this algorithm for use in common interest games and give high probability guarantees of convergence to team optimal policies in that setting.

Like the preceding works, this paper presents an independent learning algorithm for many state stochastic games and comes with convergence guarantees. However, this paper differs from those works in several ways. First, the class of games for which our learning algorithm has formal guarantees is distinct from those classes previously mentioned; at present, no algorithm comes with proven guarantees for general $N$-player symmetric games. Second, this paper also studies policy dynamics in games at large, beyond the learning setting. The structural results on $\epsilon$-satisficing paths are of independent interest, and may be of use to other algorithm designers or to those studying equilibrium computation in stochastic games.

\subsection*{Organization} The remainder of the paper is organized as follows: Section \ref{sec:model} describes the stochastic games model and presents background results. Section \ref{sec:epsilon-paths} introduces $\epsilon$-satisficing paths and proves structural results for symmetric $N$-player games and general two-player games. Building on these structural results, in Sections~\ref{sec:oracle} and \ref{sec:full-algorithm}, we develop an independent learning algorithm for $N$-player symmetric games and give convergence guarantees. The results of a simulation study are summarized in Section~\ref{sec:simulation}. Additional discussion on related and future work is given in Section~\ref{sec:discussion}. The final section concludes. Proofs omitted from the body of the text are given in the appendices. %

\subsection*{Notation} \label{ss:notation} 

$\zz_{\geq 0}$ and $\mathbb{N}$ denote the nonnegative and positive integers, respectively. For a finite set $S$, $\mathcal{P}(S)$ denotes the set of probability distributions over $S$. Given two sets $S, S'$, we let $\mathcal{P} ( S' | S)$ denote the set of stochastic kernels on $S'$ given $S$. An element $\mathcal{T} \in \mathcal{P} ( S' | S)$ is a collection of probabilities distributions on $S'$, with one distribution for each $s \in S$, and we write $\mathcal{T} ( \cdot | s )$ for $s \in S$ to make the dependence on $s$ explicit. We write $Y \sim f$ to denote that the random variable $Y$ has distribution $f$. If the distribution of $Y$ is a mixture of other distributions, say with mixture components $f_i$ and weights $p_i$ for $1 \leq i \leq n$, we write $Y \sim \sum_{i = 1}^n p_i f_i$. We use $\textbf{1} \{ \cdot \}$ to denote the indicator function of a given event. For a finite set $S$, $\uniform(S)$ denotes the uniform distribution over $S$ and $2^S$ denotes the set of subsets of $S$. %

\section{Background and Technical Preliminaries} \label{sec:model}

\subsection{Stochastic Games}
A finite stochastic game with discounted costs is described by the list
\begin{equation} \label{def:stochasticGame}
\GG = (\NN, \xx, \{ \mathbb{U}^i , c^i , \beta^i \}_{i \in \NN}, P, \nu_0 ).
\end{equation}
The components of $\GG$ are the following: $\NN$ is a finite set of $N \in \nn$ players/agents. $\xx$ is a finite set of states. For agent $i \in \NN$, $\uu^i$ is a finite set of actions, and we write $\bU := \times_{i \in \NN} \uu^i$. An element $\bu \in \bU$ is called a \emph{joint action}. For agent $i$, $c^i : \xx \times \bU \to \rr$ is a stage cost function, and $\beta^i \in [0,1)$ is a discount factor. A random initial state $x_0 \in \xx$ is given by $x_0 \sim \nu_0$ where $\nu_0 \in \PP(\xx)$. State transitions are governed by the transition kernel $P \in \PP ( \xx | \xx \times \bU )$.

At time $t \in \zz_{\geq 0}$, the state variable is denoted by $x_t \in \xx$. Each player $i \in \NN$ observes its local observation variable $y^i_t$, \by{to be described shortly,} and selects its action $u^i_t$ \by{$\in \uu^i$}. The joint action is denoted by $\bu_t = (u^i_t)_{i \in \NN} \in \bU$. Upon selection of the joint action $\bu_t$, each player $i \in \NN$ observes its realized cost $c^i ( x_t, \bu_t) \in \rr$, and the system transitions to state $x_{t+1}$, where $x_{t+1} \sim P ( \cdot | x_t, \bu_t)$.

To complete the \by{probabilistic} description of the game, we now discuss how the sequence of joint actions is generated. Each player $i \in \NN$ uses a \textit{policy} to select its sequence of actions $\{ u^i_t\}_{t \geq 0}$, using only information that is locally available at the time of each decision. We use $y^i_t$ to denote the observation variable for player $i$ at time $t \geq 0$, and we let $h^i_t$ denote player $i$'s information variable at time $t$, according to which player $i$ selects $u^i_t$.  We make the following assumption throughout this paper. 

\begin{assumption}[Independent Learners]  \label{ass:independent-learners}
For each $i \in \NN$, player $i$'s observation variables $\{ y^i_t \}_{t \geq 0}$ and information variables $\{h^i_t\}$ are given by 
\begin{itemize}
	\item $y^i_0 = x_0$ and \emph{$y^i_{t+1} = ( u^i_t, c^i ( x_t, \bu_t ), x_{t+1})$} for $t \geq 0$;
	\item $h^i_0 = y^i_0$ and $h^i_{t+1} = ( h^i_t, y^i_{t+1})$ for $t \geq 0$;
\end{itemize}
\end{assumption}

We note that this is the standard informational assumption in the literature on \textit{independent learners} \cite{zhang2021multi, matignon2012survey, matignon2009coordination, wei2016lenient,daskalakis2021independent}. The independent learner (IL) paradigm is one of the principal alternatives to the joint action learner (JAL) paradigm, studied previously in \cite{Littman1994,Littman2001ffq,Hu2003} among many others. The key difference is that the JAL paradigm assumes that each agent $i$ gets to observe the complete joint action profile $\bu_t$ after it is played; i.e. $y^i_{t+1} = ( \bu_t, c^i ( x_t, \bu_t), x_{t+1})$. In contrast, in the IL paradigm, player $i$ does not view $\bu_t$ directly at any time. 

\begin{definition}[Policies]
For player $i \in \NN$, define $\yy^i := \uu^i \times \rr \times \xx$ and $H^i_t := \xx \times ( \yy^i )^t $ for $t \geq 0$. A sequence $\pi^i = \{ \pi^i_t \}_{t \geq 0}$ of stochastic kernels is called a \emph{policy for $i$} if $\pi^i_t \in \PP ( \uu^i | H^i_t )$ for every $t \geq 0$. The set of policies for $i$ is denoted by $\Gamma^i$. 
\end{definition}

\by{ 
When following the policy $\pi^i$, agent $i$ selects its action $u^i_t$ by sampling $u^i_t \sim \pi^i_t ( \cdot | h^i_t )$. Although agents can, in principle, use arbitrarily complicated, history-dependent policies to select their actions, we will restrict our analysis to the subset of stationary policies, defined below. Such a restriction entails no loss in optimality for a particular player, provided the remaining players use stationary policies. Focusing on stationary policies is quite natural: we refer the reader to \cite{levy2013discounted} for an excellent elaboration. 
}

\begin{definition}
A policy $\pi^i \in \Gamma^i$ is called \emph{stationary} if the following holds for any $t, k \geq 0$: if \emph{$h^i_t = (x_0, u^i_0, c^i ( x_0, \bu_0 ), \cdots, x_{t-1},  u^i_{t-1}, c^i (x_{t-1}, \bu_{t-1}), x_t ) \in H^i_t$} and \emph{$\tilde{h}^i_k = (\tilde{x}_0, \tilde{u}^i_0, c^i ( \tilde{x}_0, \tilde{\bu}_0), \cdots, \tilde{x}_{k-1}, \tilde{u}^i, c^i ( \tilde{x}_{k-1}, \tilde{\bu}_{k-1}), \tilde{x}_k) \in H^i_k$} are such that \emph{$x_t = \tilde{x}_k$}, then $\pi^i_t \left( \cdot | h^i_t \right) = \pi^i_k \big( \cdot | \tilde{h}^i_k \big)$. 
\end{definition}

\by{In words, a stationary policy selects each action according to a probability distribution that depends only on the present state and not on the history or the time index.} For player $i \in \NN$, we denote the set of stationary policies by $\Gamma^i_S$ and identify $\Gamma^i_S$ with $\PP ( \uu^i | \xx)$. To ease the notational burden, in the sequel, we treat stationary policies for player $i$ as stochastic kernels on $\uu^i$ given $\xx$, without reference to the complete information or history variables. \by{Henceforth, unqualified reference to a policy shall be understood to mean a stationary policy.}

We use boldface characters to denote joint objects, such as $\bu_t = ( u^i_t)_{ i \in \NN}$ above. To isolate the role of a particular player $i \in \NN$, a joint object with $i$'s component removed is written using $-i$ in the agent index, e.g. $\bu^{-i} = (u^j )_{ j \in \NN, j \not= i}$. 

Given a joint policy $\bpi$, we use $\P^{\bpi}$ to denote the resulting probability measure on trajectories $\{ (x_t, u_t) \}_{t \geq 0}$ and we use $E^{\bpi}$ to denote the associated expectation.\footnote{\by{In principle, we should also introduce notation for the initial distribution in the probability measure, such as $\P^{\bpi}_{\nu_0}$. We omit such notation throughout this paper, because we typically condition on an initial state, making the dependence on the initial distribution redundant. In instances where we do not explicitly condition on an initial state, it should be understood that the stated property holds for any initial distribution.}} The objective of agent $i \in \NN$ is to find a policy that minimizes the expectation of its series of discounted  costs, given by
\begin{equation} \label{eq:discountedCost}
J^i   ( \bpi , x ) :=  E^{\bpi} \left[   \sum_{t \geq 0} (\beta^i)^t c^i\left(x_t,u^i_t, \bu^{-i}_t \right) \middle| x_0 = x \right]
\end{equation}
for all $x \in \xx$. Note that agent $i$ controls only its own policy, $\pi^i$, but its \by{objective function} is affected by the policies of the remaining agents. This motivates the following definitions.

\begin{definition}
Let $i \in \NN$, $\epsilon \geq 0$, and let $\Pi^i \subseteq \Gamma^i$. For $\bpi^{-i} \in \bGamma^{-i}$, a policy $\pi^{*i} \in \Pi^i$ is called an $\epsilon$-\emph{best-response} to $\bpi^{-i}$ \emph{over $\Pi^i$} if 
\begin{equation} 
	J^i  ( \pi^{*i}, \bpi^{-i}  ,  x )  \leq \inf_{\pi^i\in \Pi^i} J^i (\pi^{i},\bpi^{-i} ,  x )  + \epsilon, \quad \forall x \in \xx.
\end{equation}
\end{definition}

\begin{definition}
For fixed $i \in \NN$, $\epsilon \geq0$, $\Pi^i \subseteq \Gamma^i$, and $\bpi^{-i} \in \bGamma^{-i}$, we let $\BR^i_{\epsilon} ( \bpi^{-i} , \Pi^i )$ denote player $i$'s (possibly empty) set of $\epsilon$-best-responses to $\bpi^{-i}$ over $\Pi^i$. 
\end{definition}

\begin{definition}
Let $\epsilon \geq 0$. A joint policy $\bpi^{*} \in \bGamma$ constitutes an $\epsilon$-equilibrium if $\pi^{*i} \in \BR^i_{ \epsilon} ( \bpi^{*-i} , \Gamma^i )$ for every player $i \in \NN$. 
\end{definition}

When $\epsilon = 0$, a $0$-best-response is simply called a best-response and a $0$-equilibrium is called an equilibrium. When the set $\Pi^i$ over which $i$ is optimizing is clear from context (typically, $\Pi^i = \Gamma^i_{S}$), we may omit ``over $\Pi^i$'' and simply write $\BR^i_{\epsilon}  ( \bpi^{-i})$.

\by{
For $\epsilon \geq 0$, we let $\eq[\epsilon]_{S}$ denote the set of $\epsilon$-equilibrium policies. It is well-known that, for any finite stochastic game with discounted costs, the set of 0-equilibrium policies is non-empty \cite{fink1964equilibrium}.
}

\vspace{5pt}
\noindent The following definition will be useful in the coming sections.
\begin{definition}
Let $\xi > 0$ and $i \in \NN$. A stationary policy $\pi^i \in \Gamma^i_{S}$ is called \emph{$\xi$-soft} if $\pi^i ( a^i | x ) \geq \xi$ for every $x \in \xx$ and $a^i \in \uu^i.$ The policy $\pi^i \in \Gamma^i_{S}$ is called \emph{soft} if it is $\xi$-soft for some $\xi > 0$. 
\end{definition}

\subsection{Symmetric Games} \label{ss:symmetric-games}

In some applications, the strategic environment being modelled exhibits symmetry in the agents. To model such settings, we define a class of symmetric games with the following properties: (1) each agent has the same set of actions; (2) the state dynamics depend only on the profile of actions taken by all players, without special dependence on the identities of the agents. That is, permuting the agents' actions in a joint action leaves the conditional probabilities for the next state unchanged; (3) such a permutation results in a corresponding permutation of costs incurred. We formalize and clarify these points in the definition below. 

First, we introduce additional notation: if $\uu^i = \uu^j$ for all $i,j \in \NN$, given a permutation $\sigma : \NN \to \NN$ and joint action $\ba = (a^i)_{i \in \NN}$, we define $\sigma( \ba) \in \bU$ to be the joint action in which $i$'s component is given by $ a^{\sigma(i)}$. That is, player $i$'s action in $\sigma(\ba)$ is given by player $\sigma(i)$'s action in $\ba$: $\sigma( \ba )^i = a^{\sigma(i)}$.

\begin{definition}[Symmetric Game] \label{def:symmetric game}
A stochastic game $\GG$ given by \eqref{def:stochasticGame} is called \emph{symmetric} if the following holds:
\begin{itemize} 
	\item There exists a set $\uu$ and a constant $\beta \in ( 0,1)$ such that $\uu^i = \uu$ and $\beta^i = \beta$ for all $i \in \NN$; 
	\item For any $i \in \NN$, permutation $\sigma$ on $\NN$, and \emph{$( x , \ba ) \in \xx \times \bU$}, we have 
		\[
		c^i ( x, \sigma( \emph{\ba}) ) = c^{\sigma(i)} ( x, \emph{\ba} ) , \quad \text{and} \quad P \left( \cdot | x, \emph{\ba}   \right) = P \left( \cdot | x , \sigma(\emph{\ba} )  \right). 
		\]
\end{itemize}
\end{definition}

Observe the following useful facts about symmetric games.
\begin{lemma} \label{lemma:symmetric-equality}
Let $\GG$ be a symmetric game and let $\bpi \in \bGamma_{S}$ be a stationary joint policy. For $i, j \in \NN$, if $\pi^i  = \pi^j$, then \by{$J^i ( \pi^i, \bpi^{-i}  ,  x ) = J^j (\pi^j, \bpi^{-j}  , x )$, for any $x \in \xx.$}
%
\end{lemma}

\begin{proof} 
Let $\sigma$ be the permutation on $\NN$ such that $\sigma(i)= j$, $\sigma(j) = i$, and $\sigma(p) = p$ for all $p \in \NN \setminus \{i, j \}$. For any $t \geq 0$ and joint action $\ba \in \bU$, it follows from $\pi^i = \pi^j$ that $\P^{\bpi} ( \bu_t = \ba) = \P^{\bpi} ( \bu_t  = \sigma ( \ba ) )$. Then, since $c^{i} ( x, \sigma( \ba ) ) = c^j ( x_, \ba )$ for any state $x \in \xx$, we have the following:
\begin{align*}
&E^{\bpi} \left[ \beta^t c^i ( x_t , \bu_t ) \middle| x_t  = x \right] = \sum_{\ba \in \bU } \P^{\bpi} \left( \bu_t = \sigma ( \ba ) \middle| x_t  = x  \right) \beta^t c^i ( x, \sigma (\ba)) \\
									&= \sum_{\ba \in \bU } \P^{\bpi} \left( \bu_t =  \ba \middle| x_t = x  \right) \beta^t c^j ( x,  \ba ) = E^{\bpi} \left[ \beta^t c^j ( x_t , \bu_t ) \middle| x_t = x \right].
\end{align*}
It follows that $E^{\bpi} [ \beta^t c^i ( x_t, \bu_t ) ] = E^{\bpi} [ \beta^t c^j ( x_t, \bu_t )]$, and the result follows by summing over times $t \geq 0$ and taking limits. 
\end{proof}

\begin{corollary} \label{corollary:symmetric-equality} 
Let $\GG$ be a symmetric game and let $\bpi \in \bGamma_{S}$ be a stationary joint policy. For $i, j \in \NN$, if $\pi^i  = \pi^j$, then 
\[
\pi^i \in \BR^i_{\epsilon} ( \bpi^{-i}, \Gamma^i )  \iff \pi^j \in \BR^j_{\epsilon} ( \bpi^{-j}, \Gamma^j ) 
\]
\end{corollary}

\subsection{Background on Learning Algorithms}

\subsubsection{Learning in MDPs} \label{ss:learning-in-MDPs}

In independent learning settings, pertinent information for policy selection is not available to the players. Player $i$ does not know the policy used by players $-i$, the value of its current policy against those of the other players, or whether its current policy is an $\epsilon$-best-response. We now review how Q-learning can be used to address these uncertainties. 

Markov decision processes (MDPs) can be viewed as a stochastic game with one player, i.e. $\vert \NN \vert = 1$. In standard Q-learning  \cite{Watkins89}, a single agent interacts with its MDP environment using some policy and maintains a vector of Q-factors, the $t^{th}$ iterate denoted $Q_t \in \rr^{\xx \times \uu}$. Upon selecting action $u_t$ at state $x_t$ and observing the subsequent state $x_{t+1} \sim P ( \cdot | x_t, u_t) $ and cost $c ( x_t, u_t)$, the Q-learning agent updates its Q-factors as follows:
\begin{equation} \label{eq:q-factors}
Q_{t+1} ( x_t, u_t ) = Q_t ( x_t, u_t ) + \theta_t (x_t, u_t ) \left[ c (x_t, u_t ) + \beta \min_{a \in \uu} Q_t ( x_{t+1} , a ) - Q_t ( x_t, u_t ) \right]
\end{equation}
where $\theta_t (x_t, u_t) \in [0,1]$ is a random step-size parameter and $Q_{t+1} ( s, a) = Q_t (s,a)$ for all $(s,a) \not= (x_t, u_t)$.\footnote{We are interested in the tabular, online variant, where access to the state, action, and cost feedback arrive piece-by-piece as the agent interacts with its environment.}

Under mild conditions, $Q_t \to Q^*$ almost surely as $t \to \infty$, where $Q^* \in \rr^{\xx \times \uu}$ is variously called the (state-)action value function or the Q-function \cite{WatkinsDayan92,tsitsiklis1994asynchronous}. The value $Q^*(s,a)$ represents the expected discounted cost-to-go from the initial state $s$, assuming that the agent initially chooses action $a$ and follows an optimal policy thereafter. The function $Q^{*}$ is given by 
\[
Q^{*}(s,a) = E^{\pi^{*}} \left[ \sum_{t = 0}^{\infty} \beta^i c ( x_t, u_t ) \middle| x_0 = s, u_0 = a \right] \quad \forall (s,a) \in \xx \times \uu, 
\]
where $\pi^{*}$ is any optimal policy. The vector $Q^*$ can then be used to construct any optimal policy $\tilde{\pi}^*$ in a state-by-state manner by setting
\[
\tilde{\pi}^* ( a^{*} | x) = 1, \text{ where } a^{*} \in \left\{ u \in \uu : Q^* ( x, u ) = \min_{a \in \uu} Q^* (x, a) \right\} \quad \forall x \in \xx. 
\]

\subsubsection{Learning in Stochastic Games} \label{ss:learning-in-games}

In the single-agent literature, the MDP is fixed and the $Q^*$ notation is used, but one could also introduce notation to identify the MDP. Returning to the game setting, if all agents except $i$ follow a stationary policy $\bpi^{-i} \in \bGamma^{-i}_{S}$, agent $i$ faces an environment that is equivalent to an MDP that depends on $\bpi^{-i}$. We denote agent $i$'s $t^{th}$ Q-factor iterate by $Q^i_t$ and $i$'s optimal Q-factors when playing against $\bpi^{-i}$ by $Q^{*i}_{\bpi^{-i}} \in \rr^{\xx \times \uu^i}$. With this notation, $Q^{*i}_{\bpi^{-i}}(x,u^i)$ represents agent $i$'s expected discounted cost-to-go from the initial state $x$ assuming that agent $i$ initially chooses $u^i$ and uses an optimal policy thereafter while the other agents use $\bpi^{-i}$, a fixed stationary policy. We note that an optimal policy for $i$ is guaranteed to exist since $i$ faces a finite, discounted MDP, and that any optimal policy for $i$ in this MDP is a $0$-best-response to $\bpi^{-i}$ in the underlying game $\GG$. More generally, we have the following fact: for any $i \in \NN$, $\bpi^{-i} \in \bGamma^{-i}_{S}$,
\[
\pi^i \in \BR^i_{\epsilon} ( \bpi^{-i}, \Gamma^i_{S} ) \iff J^i ( \pi^i , \bpi^{-i} ,  x ) \leq \min_{a^i \in \uu^i } Q^{*i}_{\bpi^{-i}} ( x, a^i ) + \epsilon , \quad \forall x \in \xx. 
\]

\subsection{Continuity of Value Functions and Quantized Policies}

We now present some useful results on the continuity of the various value functions introduced above. We begin by metrizing the policy sets $\bGamma_{S}$ and $\Gamma^i_{S}$ for each player $i \in \NN$. For $i \in \NN$, we define a metric $d^i$ on $\Gamma^i_{S}$ by setting 
\[
d^i ( \pi^i, \tilde{\pi}^i ) := \max \left\{ \left| \pi^i ( a^i | s ) - \tilde{\pi}^i ( a^i | s ) \right| : s \in \xx, a^i \in \uu^i \right\}, \quad \forall \pi^i, \tilde{\pi}^i \in \Gamma^i_{S}. 
\]

We then define the metric $\bd$ on $\bGamma_{S}$ by $\bd ( \bpi, \tilde{\bpi} ) := \max_{i \in \NN} d^i ( \pi^i, \tilde{\pi}^i ).$ The sets of policies $\{ \Gamma^i_{S} \}_{i \in \NN}$ are then compact in the topology induced by the metrics $\{ d^i \}_{i \in \NN}$, and similarly $\bGamma_{S}$ is compact in the topology induced by $\bd$.

\begin{lemma} \label{lemma:continuous-cost} 
For any player $i \in \NN$ and state $s \in \xx$, the function $\varphi^{i}_{s} : \bGamma_{S} \to \rr$ given by
\[
\varphi^{i}_{s} ( \bpi ) = J^i ( \bpi , s ) \quad \forall \bpi \in \bGamma_{S}
\]
is continuous. 
\end{lemma}

\begin{lemma} \label{lemma:continuous-Q-factors}
For any player $i \in \NN$ and state-action $(s, a^i)\in \xx\times\uu^i$, the mapping $\phi^{i}_{(s,a^i)} : \bGamma_{S}^{-i} \to \rr$ given by 
\[
\phi^{i}_{(s,a^i)} ( \bpi^{-i} ) = Q^{*i}_{\bpi^{-i}} ( s, a^i ), \quad \forall \bpi^{-i} \in \bGamma^{-i}_{S}
\]
is continuous.

\end{lemma}

\begin{lemma} \label{lemma:continuous-min-action}
For any player $i \in \NN$ and state $s \in \xx$, we have that the mapping $f^i_s : \bGamma_{S}^{-i} \to \rr$ given by 
\[ f^i_s ( \bpi^{-i} ) = \min_{a^i \in \uu^i } Q^{*i}_{\bpi^{-i}} ( s, a^i ) , \quad \forall \bpi^{-i} \in \bGamma_{S}^{-i}
\] is continuous. 
\end{lemma}

\begin{lemma} \label{lemma:continuous-max-state}
For any player $i \in \NN$ and fixed policy $\bpi^{-i} \in \bGamma_{S}^{-i}$, we have that the mapping $g^i_{\bpi^{-i}} : \Gamma_{S}^i \to \rr$ given by 
\[
 g^i_{\bpi^{-i}} (\pi^i ) = \max_{ s \in \xx } \left( J^i (\pi^i, \bpi^{-i} ,  s ) - \min_{a^i \in \uu^i} Q^{*i}_{\bpi^{-i}} (s, a^i)		\right) , \quad \forall \pi^i \in \Gamma^{i}_{S}
\]
is continuous. 
\end{lemma}

The proofs of Lemmas~\ref{lemma:continuous-cost}--\ref{lemma:continuous-max-state} can be found in Appendix~\ref{appendix:proofs-of-continuity-lemmas}.

\subsubsection{Quantizing Policy Sets} For algorithm design purposes, we now consider the effects of restricting each player to select its policy from a finite subset of its stationary policies. The finite subset of policies will be obtained be a fine quantization of the set of all stationary policies. Due to the preceding results on continuity of the many value functions, if this quantization is sufficiently fine, then restriction to this finite subset of policies entails only a small loss in performance.

\begin{definition}
Let $\xi > 0$, and $i \in \NN$. A mapping $q^i : \Gamma^i_{S} \to \Gamma^{i}_{S}$ is called a \emph{$\xi$-quantizer} if 
\begin{itemize}
	\item[(i)] the set $q^i ( \Gamma^i_{S} ) = \{ q^i ( \pi^i ) : \pi^i \in \Gamma^{i}_{S} \}$ is finite, and 
	\item[(ii)] For all $\pi^i \in \Gamma^i_{S}$, we have that $d^i ( \pi^i, q^i ( \pi^i )) < \xi$.  
\end{itemize}
\end{definition}

\by{We remark that if $0 < \xi_1 < \xi_2$, then any $\xi_1$-quantizer is automatically a $\xi_2$-quantizer, by part (ii) in the preceding definition.}

\begin{definition}
Let $i \in \NN$ and $\xi > 0$. A subset $\Pi^i \subseteq \Gamma^i_{S}$ is called a \emph{$\xi$-quantization} of $\Gamma^i_{S}$ if $\Pi^i = q^i ( \Gamma^i_{S} )$ for some $\xi$-quantizer $q^i$.
\end{definition}

We extend this terminology to also refer to subsets $\bPi \subseteq \bGamma_{S}$ as $\xi$-quantizations of $\bGamma_{S}$ when, for each component $i \in \NN$, $\Pi^i$ is a $\xi$-quantization of $\Gamma^i_{S}$. We note that for any $\xi > 0$ and $i \in \NN$, by the compactness of $\Gamma^i_{S}$, there always exists some $\xi$-quantization $\Pi^i$ of $\Gamma^i_{S}$ such that all policies $\pi^i \in \Pi^i$ are soft.

Since the sets $\{ \Gamma^i_{S} \}_{i \in \NN}$ and $\bGamma_{S}$ are compact, it follows from Lemma~\ref{lemma:continuous-cost} that the cost functionals are uniformly continuous on $\bGamma_{S}$. That is, for any $\delta > 0$, there exists $\xi = \xi ( \delta )$ such that for any $i \in \NN$, $x \in \xx$, and joint policies $\bpi, \tilde{\bpi} \in \bGamma_{S}$, if $\bd ( \bpi, \tilde{ \bpi } ) < \xi$, then we have $ \left| J^i ( \bpi, x) - J^i ( \tilde{\bpi} , x) 	\right| < \delta .$

\by{For $\epsilon > 0$,} a quantization $\bPi$ of $\bGamma_{S}$ into bins of radius less than \by{$\xi ( \frac{\epsilon}{3})$} has the desirable property that the quantization $\Pi^i$ always has contains an $\by{\frac{\epsilon}{3}}$-best-response to any policy $\bpi^{-i} \in \bGamma^{-i}_{S}$ of the remaining players. That is, $\Pi^i \cap \BR^i_{\by{\epsilon/3}} ( \bpi^{-i} , \Gamma^i_{S} ) \not= \varnothing$. Moreover, we are guaranteed at least one $\epsilon$-equilibrium in the quantization $\bPi$.

\begin{corollary} \label{corollary:quantized-equilibrium}
For any $\epsilon > 0$, there exists $\xi = \xi ( \epsilon, \GG ) > 0$ such that if $\bPi$ is a $\xi$-quantization of $\bGamma_{S}$, then we have $\bPi \cap \eq[\epsilon]_{S} \not= \varnothing$. 
\end{corollary}

\by{We note that, by a previous remark, Corollary~\ref{corollary:quantized-equilibrium} holds for any $\xi^{\prime}$-quantization of $\bGamma_{S}$, where $\xi^{\prime} < \xi( \epsilon , \GG )$. Furthermore, such a quantization can always be selected so as to only contain soft policies.}

\section{Policy Dynamics and $\epsilon$-Satisficing} \label{sec:epsilon-paths}

This section studies discrete-time dynamical systems on the set of stationary joint policies $\bGamma_{S}$. In particular, we focus on those dynamical systems whose trajectories are obtained by a \textit{policy revision process}, in which each agent changes its policy according some update rule. While we do not restrict ourselves to studying systems in which all agents use the same update rule, we do focus on rules of the so-called $\epsilon$-satisficing variety, to be defined shortly. Broadly speaking, \by{we} wish to address the following question: when can it be guaranteed that \textit{some} $\epsilon$-satisficing policy revision process drives play to $\epsilon$-equilibrium irrespective of the initial joint policy? We present positive results for $N$ player symmetric games and general two player games.

For the following definitions, we let $\GG$ be a stochastic game given by \eqref{def:stochasticGame}.

\begin{definition}
For player $i \in \NN$, a function $T^i : \bGamma_{S} \to \Gamma^i_{S}$ is called a \textit{policy update rule for player $i$}. 
\end{definition}

Given a collection $\bT = \{ T^i : i \in \NN \}$ of policy update rules for each player, we will study the discrete-time orbits of $\bT$. For $\bpi \in \bGamma_{S}$, we define $\bT ( \bpi)$ to be the stationary joint policy where $i$'s component is given by $T^i (\bpi)$. That is, $\bT ( \bpi ) := ( T^i ( \bpi ) )_{i \in \NN}$. Furthermore, for all $\bpi \in \bGamma_{S}$, we put $\bT^{0} ( \bpi ) = \bpi$ and 
\[
\bT^{k+1} ( \bpi ) = \bT ( \bT^k ( \bpi )) , \quad \forall k \geq 0. 
\]

\begin{definition}
Let \emph{$\bT = \{ T^i : i \in \NN \}$} be a collection of policy update rules for each player, and let $\tt : \bGamma_{S} \to \{  ( \bpi_k )_{k \geq 0} : \bpi_k \in \bGamma_{S} \text{ for all } k \geq 0 \}$ be a mapping from joint policies to sequences of joint policies. 

If $\tt ( \bpi ) = (    \emph{\bT}^k ( \bpi ) )_{k \geq 0}$ for every $\bpi \in \bGamma_{S}$, we say that $\tt$ is the \emph{policy revision process associated to $\bT$.}
\end{definition}

Policy revision processes arising from specific update rules have received considerable attention in the past, both in discrete time and continuous time settings. Special interest has been given to best-response dynamics and its variants (see, for instance, \cite{leslie2020best,matsui1992best,roughgarden2010algorithmic}), replicator dynamics \cite{gaunersdorfer1995fictitious,hofbauer2003evolutionary}, and the relationship between these dynamics and game theoretic learning. 

Relatively fewer works focus on entire classes of policy revision processes. Two pioneering works in this tradition are \cite{hart2003uncoupled} and \cite{hart2006stochastic}. In \cite{hart2003uncoupled}, the authors consider continuous time policy revision processes in normal form games and give a non-existence result: if the policy update rules satisfy certain regularity properties as well as a condition called uncoupledness, then the associated policy revision process may not converge to Nash equilibrium. In \cite{hart2006stochastic}, the authors consider discrete time stochastic policy revision processes arising from uncoupled dynamics and provide a positive result, contingent on the policy update rules also allowing for use of memory.

Like \cite{hart2003uncoupled} and \cite{hart2006stochastic}, we will study a class of policy revision processes, rather than focusing on a revision process associated to a particular update rule. Instead of focusing on uncoupled dynamics, however, we will study policy update rules that instruct an agent to keep using its current policy whenever that policy is an $\epsilon$-best-response to the prevailing joint policy. Such update rules, which we formalize as $\epsilon$-satisficing update rules below, have the desirable property that $\epsilon$-equilibrium policies are stable points for the associated policy revision processes.

\begin{definition}
Let $\epsilon \geq 0$ and let $T^i$ be a policy update rule for player $i \in \NN$. $T^i$ is said to be \emph{$\epsilon$-satisficing} if, for any $( \pi^i, \bpi^{-i} ) \in \bGamma_{S}$, we have that $\pi^i \in \BR^i_{\epsilon} ( \pi^i, \bpi^{-i} )$ implies $T^i ( \bpi ) = \pi^i.$
\end{definition}

Our chosen terminology is inspired by \cite{simon1956rational}, where ``satisficing'' refers to becoming satisfied and halting search when a sufficiently good input has been found in an optimization problem. Satisficing has a long history in both single-agent decision theory (e.g. \cite{radner1975satisficing}, \cite{cassidy1972solution}) and also multi-agent game theory (e.g. \cite{charnes1963deterministic}). Recently, there has been renewed interest in studying dynamics arising from particular $\epsilon$-satisficing policy update rules; for example, see \cite[Section 5]{candogan2013near} and \cite{chien2011convergence}.

A policy revision process $\tt$ is called $\epsilon$-satisficing if it is associated to a collection of policy update rules $\bT = \{ T^i \}_{i \in \NN}$ such that $T^i$ is $\epsilon$-satisficing for each player $i \in \NN$. It is natural to ask the following: what assumptions must be made on a game in order to guarantee that \textit{some} $\epsilon$-satisficing revision process can drive the joint policy process to $\eq[\epsilon]_{S}$ irrespective of the initial policy? With this question in mind, we state the following definitions.

\begin{definition} \label{def:epsilon-path}
Let $\epsilon \geq 0$. A (possibly finite) sequence $( \bpi )_{k \geq 0}$ of stationary joint policies is called an \emph{$\epsilon$-satisficing path} if, for every $k \geq 0$ and $i \in \NN$, $\pi^i_k \in \BR^i_{\epsilon} ( \bpi^{-i}_k )$ implies $\pi^i_{k+1} = \pi^i_k$.
\end{definition}

\begin{definition} \label{def:paths-property}
Let $\epsilon \geq 0$ and let $\bPi \subseteq \bGamma_{S}$ be a subset of stationary joint policies. A game $\GG$ is said to have the $\epsilon$-satisficing paths property within $\bPi$ if for every $\bpi \in \bPi$, there exists an $\epsilon$-satisficing path $( \bpi_t )_{t \geq 0}$ and an integer $K = K(\bpi)$, such that (i)~$\bpi_0=\bpi$, (ii) $\bpi_t \in \bPi$ for every $t \geq 0$, and (iii) $\bpi_K \in \eq[\epsilon]_{S}$.  
\end{definition}

We note that Definitions~\ref{def:epsilon-path} and \ref{def:paths-property} are not attached to any particular policy revision process or collection of policy update rules. In particular, we note that Definition~\ref{def:epsilon-path} does not require that a player must switch to a best-response when not already $\epsilon$-best-responding. Consequently, one may interpret the $\epsilon$-satisficing paths property as a necessary condition for convergence to $\eq[\epsilon]_S$ when players use arbitrary $\epsilon$-satisficing update rules. It is therefore useful to establish whether this property holds (or fails to hold): in games where the $\epsilon$-satisficing paths property does not hold within $\bGamma_{S}$, the use of $\epsilon$-satisficing policy update rules may be inappropriate, as there exist initial joint policies from which $\eq[\epsilon]_{S}$ cannot be reached in finite time by following an $\epsilon$-satisficing path.

\subsection{Satisficing Paths in $N$-Player Symmetric Games}
\begin{theorem} \label{thm:symmetric-paths} 
	Let $\GG$ be a symmetric stochastic game given by \eqref{def:stochasticGame}. Then $\GG$ has the $\epsilon$-satisficing paths property in $\bGamma_{S}$ for all $\epsilon \geq 0$. 
\end{theorem}

In the proof below, we construct an $\epsilon$-satisficing path of finite length from any initial policy into the set $\eq[\epsilon]_{S}$. Intuitively, beginning from an arbitrary policy, unsatisfied players (i.e. players not $\epsilon$-best-responding) can change policies to match the policy of some other (not necessarily satisfied) player. We create a cohort of players using the same policy and progressively grow the cohort---either by adding an unsatisfied player to the cohort by switching its policy, or by switching the policy of every member of the cohort to match that of some other player---until a stopping condition is met. We stop either because we have found an $\epsilon$-equilibrium or because no player is satisfied, which allows us to move in one step to an arbitrary $\epsilon$-equilibrium. 
\begin{proof}
Let $\bpi_0 \in \bGamma_{S}$ be an initial policy. We claim that there exists some $\epsilon$-satisficing path of finite length from $\bpi_0$ to $\eq[\epsilon]_{S}$. Put $C_{-1} = \varnothing$ and select a player $i(0) \in \NN$ arbitrarily. We define our first cohort, $C_0$ to be the subset of players whose policy matches that of player $i(0)$: $C_0 := \{ j \in \NN : \pi^j_0 = \pi^{i(0)}_0 \}.$
 
For some $n \geq 0$, suppose that we have a sequence of joint policies $\{ \bpi_k \}_{k = 0}^n$ and player subsets $\{ C_k \}_{k = 0}^n$ such that items (1)--(4) below hold for each $k \in \{0 , \dots, n \}$:
	\begin{itemize}
		\item[(1)] All players in $C_k$ use the same policy, i.e. $\pi^i_k = \pi^j_k$ for all $i, j \in C_k$;
		\item[(2)] $C_{k-1} \subset C_k$ and $| C_k | \geq | C_{k-1} | + 1$;
		\item[(3)] If player $i \in C_k$, $j \notin C_k$, then $\pi^j_k \not= \pi^i_k$;
		\item[(4)] $\bpi_0, \cdots, \bpi_k$ is an $\epsilon$-satisficing path. 
	\end{itemize}

\ \\
If $\bpi_n \in \eq[\epsilon]_{S}$, then $(\bpi_0 , \cdots, \bpi_n)$ is an $\epsilon$-satisficing path of finite length from $\bpi_0$ to $\eq[\epsilon]_S$. If $\bpi_n \notin \eq[\epsilon]_{S}$ and $C_n = \NN$, then by Corollary~\ref{corollary:symmetric-equality}, $\pi^i_n \notin \BR^i_{\epsilon} ( \bpi^{-i}_n)$ for all $i \in \NN$, and so all players may change their policies. It follows that for any $\bpi^{*} \in \eq[\epsilon]_{S}$, the sequence $(\bpi_0, \cdots, \bpi_n, \bpi^{*} )$ is an $\epsilon$-satisficing path of finite length from $\bpi_0$ to $\eq[\epsilon]_S$. %

We now focus on the final case: $\bpi_n \notin \eq[\epsilon]_{S}$ and $C_n \not= \NN$. In this case, the sequence $( \bpi_k , C_k )_{k = 0}^{n}$ is submaximal, in that there exists a policy $\bpi_{n+1}$ and a player subset $C_{n+1}$ such that the extended sequence $( \bpi_k , C_k )_{k = 0}^{n+1}$ satisfies (1)--(4) for each $k \in \{ 0, \dots, n+1 \}$. We produce such a policy $\bpi_{n+1}$ and player set $C_{n+1}$ now, proceeding in cases. 

Since $\bpi_{n} \notin \eq[\epsilon]_{S}$, there exists some player who is not $\epsilon$-best-responding at $\bpi_n$, i.e. there exists $i \in \NN$ such that $\pi^i_n \notin \BR^i_{\epsilon} ( \bpi^{-i}_n)$. We have two sub-cases to consider: by Corollary~\ref{corollary:symmetric-equality}, either (a) $\pi^j_n \in \BR^j_{\epsilon} ( \bpi^{-j}_n)$ for every $j \in C_n$, or (b)  $\pi^j_n \notin \BR^j_{\epsilon} ( \bpi^{-j}_n)$ for every $j \in C_n$.

In sub-case (a), all players in $C_n$ are $\epsilon$-best-responding. Select an unsatisfied player $j(n+1) \in \NN \setminus C_n$, and construct a successor policy $\bpi_{n+1}$ by putting $\pi^{j(n+1)}_{n+1} = \pi^{i^*}_n$, where $i^{*} \in C_n$ is any player in $C_n$, and put $\pi^i_{n+1} = \pi^i_n$ for all players $i \not= j(n+1)$. We then put $C_{n+1} = C_n \cup \{ j(n+1) \}$, and the sequence $\{ \bpi_k , C_k \}_{k = 0}^{n}$ is extended to $\{ \bpi_k , C_k \}_{k = 0}^{n+1}$ while preserving (1)--(4) for all $k \in \{0, \dots, n+1\}$.

In sub-case (b), all players in $C_n$ are allowed to switch their policy while preserving the $\epsilon$-satisficing property of the path. Select a player $h(n+1) \in \NN \setminus C_n$ and define the successor policy $\bpi_{n+1}$ as follows:
\begin{align*}
\pi^i_{n+1} =	\begin{cases}
				\pi^i_{n} &\text{ if } i \notin C_n \\ 
				\pi^{h (n+1)}_{n} &\text{ if } i \in  C_n.
			\end{cases} 
\end{align*}
We note that the player $h(n+1)$ may be selected arbitrarily from $\NN \setminus C_n$ and need not be $\epsilon$-best-responding to $\bpi^{-h(n+1)}_{n}$. Next, we define $C_{n+1} = C_n \cup \{ j \in \NN : \pi^j_n = \pi^{h( n+1)}_n \}$. Thus, the sequence $\{ \bpi_k , C_k \}_{k = 0}^{n}$ has been extended to $\{ \bpi_k , C_k \}_{k = 0}^{n+1}$ while preserving (1)--(4) for all $k \in \{1, \dots, n+1\}$. 

The preceding extension process---of obtaining policy $\bpi_{n+1}$ and cohort $C_{n+1}$ from $\bpi_n$ and $C_n$---can be repeated at most finitely many times before one of the two aforementioned stopping conditions  (namely, $C_{n+1} = \NN$ or $\bpi_{n+1} \in \eq[\epsilon]_{S}$) is met. If the stopping condition $\bpi_{n+1} \in \eq[\epsilon]_{S}$ is satisfied, we have produced an $\epsilon$-satisficing path of finite length from $\bpi_0$ into $\eq[\epsilon]_{S}$. Otherwise, the stopping condition $C_{n+1} = \NN$ is satisfied while $\bpi_{n+1} \notin \eq[\epsilon]_{S}$, in which case all players may switch policies and $( \bpi_0 , \cdots, \bpi_{n+1},  \bpi^{*} )$ is an $\epsilon$-satisficing into $\eq[\epsilon]_{S}$ for any $\bpi^{*} \in \eq[\epsilon]$. 
\end{proof}

\ \\
In fact, the argument in the proof of Theorem \ref{thm:symmetric-paths} can be used, without modification, to prove the following result, in which the policy set is restricted.

\begin{theorem} \label{thm:generalized-symmetric-paths}
Let $\GG$ be a symmetric game given by \eqref{def:stochasticGame} and let $\epsilon \geq 0$. Let $\bPi \subseteq \bGamma_{S}$ be a subset of stationary joint policies satisfying $\Pi^i = \Pi^j$ for all $i, j \in \NN$ and suppose $\bPi \cap \eq[\epsilon]_{S} \not= \varnothing$. Then, $\GG$ has the $\epsilon$-satisficing paths property within $\bPi$. 
\end{theorem}

\by{Although this result may initially appear to be only a modest generalization of Theorem~\ref{thm:symmetric-paths}, we will see in the next section that it has important consequences for algorithm design. In particular, we will see that the existence of $\epsilon$-satisficing paths within a finite subset of policies is a sufficient condition for the convergence of some learning processes in symmetric games.}

\subsection{Satisficing paths in General Two-Player Games}

In this subsection, we state and prove our second structural result, which is that general two-player stochastic games have the $\epsilon$-satisficing paths property for any $\epsilon > 0$. This result assumes no symmetry in the game, and therefore requires a rather different proof technique than the one used for Theorem~\ref{thm:symmetric-paths}. The proof used here is non-constructive and relies on the continuity properties of various value functions.

\begin{theorem}  \label{thm:two-player-paths}
Let $\GG$ be a stochastic game given by \eqref{def:stochasticGame} with $| \NN | = 2$. Then, $\GG$ has the $\epsilon$-satisficing paths property within $\bGamma_{S}$ for any $\epsilon > 0$.
\end{theorem}

\begin{proof}
Fix $\bpi_0 \in \bGamma_{S}$, and let $\sat_{\epsilon} ( \bpi_0) := \{ i \in \NN : \pi^i_0 \in \BR^i_{\epsilon} ( \bpi^{-i}_0 ) \}$. We will argue that there exists an $\epsilon$-satisficing path $(\bpi_0, \bpi_1, \bpi_2)$ such that $\bpi_2 \in \eq[\epsilon]_{S}$.

We proceed in three cases: either (1) $\left| \sat_{\epsilon} (\bpi_0) \right| = 0$, (2) $\left| \sat_{\epsilon} (\bpi_0) \right| = 1$, or (3) $\left| \sat_{\epsilon} (\bpi_0) \right| $ = 2. Cases (1) and (3) are straightforward---in Case (1) select any $\bpi^{*} \in \eq[\epsilon]_{S}$ and put $\bpi_1 = \bpi_2 = \bpi^{*}$; in Case (3), put $\bpi_0 = \bpi_1 = \bpi_2$. Then, in either case we have that $(\bpi_0, \bpi_1, \bpi_2)$  is an $\epsilon$-satisficing path that $\bpi_2 \in \eq[\epsilon]_{S}$.

We focus now on Case (2), where exactly one player is $\epsilon$-best-responding initially. Let $i \in \sat_{\epsilon} ( \bpi_0)$ denote the player who is $\epsilon$-best-responding, and let $j \in \NN \setminus \sat_{\epsilon} ( \bpi_0 )$ denote the player who is not $\epsilon$-best-responding at $\bpi_0$. There are two sub-cases to consider: 
\begin{itemize}
	\item[(a)] There exists $\pi^j \in \Gamma^j_{S} \setminus \BR^j_{\epsilon} ( \pi^i_0 )$ such that $\pi^i_0 \notin \BR^i_{\epsilon} ( \pi^j )$;
	\item[(b)] For all $\pi^j \in \Gamma^j_{S} \setminus \BR^j_{\epsilon} ( \pi^i_0)$, we have that $\pi^i_0 \in \BR^i_{\epsilon} (\pi^j)$. 
\end{itemize}

\ \\ 
In case (a), suppose $\pi^j_1 \in \Gamma^j_{S} \setminus \BR^j_{\epsilon} ( \pi^i_0 )$ is such that $\pi^i_0 \notin \BR^i_{\epsilon} ( \pi^j_1 )$, and put $\bpi_1 = (\pi^i_0, \pi^j_1)$. Then, $(\bpi_0, \bpi_1)$ is an $\epsilon$-satisficing path, since only $j$ changed its policy from $\bpi_0$ to $\bpi_1$. By the condition defining case (a), neither player is $\epsilon$-best-responding at $\bpi_1$, and so both may change policies. Thus, for any $\bpi^{*} \in \eq[\epsilon]_{S}$, we have that $( \bpi_0, \bpi_1, \bpi^{*} )$ is an $\epsilon$-satisficing path into $\eq[\epsilon]_{S}$.

\ \\
For case (b), we will argue that the premise implies that there exists some $\pi^{*j}_1 \in \BR^j_{\epsilon} ( \pi^i_0)$ such that $\bpi_1 = ( \pi^i_0, \pi^{*j}_1 ) \in \eq[\epsilon]_{S}$. 

Note that for any $\pi^j \in \Gamma^j_{S}$, we have that $\pi^j \in \BR^j_{\epsilon} ( \pi^i_0 )$ if and only if 
\[
\max_{s \in \xx } \left(  J^j  (\pi^j, \pi^i_0 , s ) - \min_{a^j \in \uu^j } Q^{*j}_{\pi^i_0} ( s, a^j )		\right) \leq \epsilon ,
\]
and an analogous condition characterizes $\epsilon$-best-responding for player $i$. 

We will use this formulation of $\epsilon$-best-responding to construct a continuous function $\Phi : [0,1] \to \rr$. Fix some 0-best-response $\tilde{\pi}^j \in \BR^j_{0} ( \pi^i_0 )$. For each $\lambda \in [0,1]$, we define a policy $\pi^j_{(\lambda)} \in \Gamma^j_{S}$ as
\[
\pi^j_{(\lambda)} ( \cdot | x ) = (1-\lambda) \pi^j_0 ( \cdot | x ) + \lambda  \tilde{\pi}^j ( \cdot | x ) , \quad \forall x \in \xx. 
\]
We define $\Phi : [0,1] \to \rr$ by 
\[
\Phi ( \lambda ) := \max_{s \in \xx } \left(   J^j \left( \pi^j_{(\lambda)} , \pi^i_0 , s \right) - \min_{a^j \in \uu^j } Q^{*j}_{\pi^i_0} ( s, a^j ) \right ), \quad \forall \lambda \in [0,1]. 
\]
Note that $0 = \Phi(1) < \epsilon$, since $\pi^j_{(1)} = \tilde{\pi}^j \in \BR^j_{0} ( \pi^i_0)$, and that $\epsilon < \Phi ( 0 )$, since $\pi^j_{(0)} = \pi^j_0 \notin \BR^j_{\epsilon} (\pi^i_0)$. As it is the composition of continuous functions (Lemmas~\ref{lemma:continuous-cost}--\ref{lemma:continuous-max-state}), we have that $\Phi$ is continuous. By the intermediate value theorem, there exists $\lambda \in (0,1)$ such that $\Phi ( \lambda ) = \epsilon$. We put $\lambda^{*} = \inf \{ \lambda \in (0,1) : \Phi ( \lambda ) = \epsilon \} > 0$. It follows that for any $\lambda < \lambda^{*}$, we have $\pi^j_{ ( \lambda ) } \notin \BR^j_{\epsilon}  ( \pi^i_0)$, otherwise the minimality of $\lambda^{*}$ is contradicted. 

We take an increasing sequence $\{ \lambda_n \}_{n \geq 0}$ such that $\lambda_n \uparrow \lambda^{*}$, and we define the policies $\gamma^j_n := \pi^j_{( \lambda_n )}$ for each $n \geq 0$. We have that $\gamma^j_n \to \pi^j_{ ( \lambda^{*} ) }$, and furthermore $\gamma^j_n \in \Gamma^j_{S} \setminus \BR^j_{\epsilon} (\pi^i_0)$ for each $n \geq 0$, while $\pi^j_{ (\lambda^{*} )} \in \BR^j_{\epsilon} ( \pi^i_0)$. 

By the condition defining case (b), we have that $\pi^i_0 \in \BR^i_{\epsilon}  ( \gamma^j_n )$ for all $n \geq 0$. Equivalently,
\[
\max_{ s \in \xx } \left( 	J^i \left( \pi^i_0 , \gamma^j_n , s \right) - \min_{ a^i \in \uu^i } Q^{*i}_{\gamma^j_n } ( s, a^i ) \right) \leq \epsilon \quad \forall n \geq 0.
\]
By continuity, this holds taking the limit as $n \to \infty$, and so $\pi^i_0 \in \BR^i_{\epsilon} ( \pi^j_{ ( \lambda^* ) } )$. Thus, in case (b), we may put $\bpi_1 = \bpi_2 = ( \pi^i_0, \pi^j_{ (\lambda^{*} ) } ) \in \eq[\epsilon]_{S}$, which completes the proof. 
\end{proof}

\by{
\subsection*{Remarks}

We now offer some intuition about the argument used to prove Theorem~\ref{thm:two-player-paths}, and discuss difficulties that arise when one attempts to generalize this proof method to $N$-player games, with $N \geq 3$, or to restricted policy subsets $\bPi \subset \bGamma_{S}$. 

Cases (1) and (3) in the preceding proof are intuitively simple. In case (1), neither agent is presently $\epsilon$-best-responding and therefore both agents may switch their policies to any successor. In case (3), both agents are presently $\epsilon$-best-responding and the existence of a path to $\epsilon$-equilibrium is trivial. This leaves case (2), where exactly one agent is $\epsilon$-best-responding, as the only remaining case.

In intuitive terms, we analyze case (2) by asking whether the unsatisfied player (player $j$ in the proof above) can destabilize the other player \textit{without making itself satisfied.} The ability to induce mutual dissatisfaction is the defining property of case (2a), and leads to a very simple analysis: if the $\epsilon$-unsatisfied player can make both players unsatisfied, then both are free to switch policies in the next period. The remaining sub-case, case (2b), is simply the logical negation of case (2a). This final case involves a satisfied player, $i$, whose satisfaction cannot be destabilized by the unsatisfied player $j$ as long as $j$ remains unsatisfied. This satisfied disposition allows the unsatisfied player to approach a best-response without unsettling the already satisfied player.

Unfortunately, generalizing this proof technique to games with more than two players is rather challenging. When considering an $N$-player game with $N > 2$, in addition to the trivial cases---where no players are $\epsilon$-satisfied and where all players are $\epsilon$-satisfied---there are $N-1$ middling cases, where there are exactly $k$ satisfied players, with $1 \leq k \leq N-1$.  When $k > 1$, the condition analogous to case (2a) remains easy to analyze, but its logical negation fails to be useful. In the sub-case analogous to case (2b), 
the set of $k$ initially satisfied players is not a monolith: changing policies may leave some of the formerly satisfied players satisfied while making others unsatisfied. As a result, the technique of taking limits and relying on continuity properties of the value functions may not yield the desired result. 

As another matter, the proof technique used to prove Theorem~\ref{thm:two-player-paths} does not readily apply to (finite) policy subsets. That is, this proof technique does not immediately lead to a generalization of Theorem~\ref{thm:two-player-paths} in the way that Theorem~\ref{thm:generalized-symmetric-paths} generalized Theorem~\ref{thm:symmetric-paths}. As we will see in the coming sections, a consequence is that algorithm design for general two-player stochastic games is rather more involved than in symmetric games with $N$-players. 
}

\section{Exploiting the $\epsilon$-Satisficing Paths Property} \label{sec:oracle}

In this section and the next, we demonstrate how one can design independent learners that exploit the $\epsilon$-satisficing paths property of symmetric games. We begin, in this section, by developing intuition in the simplified setting where learning is black-boxed and players update their policies using an $\epsilon$-satisficing rule that incorporates random search when not $\epsilon$-best-responding. A complete independent learning algorithm suitable for online learning in symmetric games is then presented in the next section and analyzed as a two-timescale, noisy implementation of the black-boxed process.

\subsection{Policy revision with oracle} 

In Algorithm~\ref{algo:oracle}, we present an procedure in which players randomly revise their policies in discrete time, resulting in a time homogenous Markov chain on $\bGamma_{S}$. There is no learning in this idealized process: at each time step, player $i \in \NN$ receives the relevant state value and action value information from an oracle, and uses this information to select its successor policy. We assume that the policy updates are jointly independent across agents, conditional on the information given by the oracle.  \by{That is, we do not assume shared randomness.}

The relevant parameters and objects used in Algorithm~\ref{algo:oracle} are the following: 
\begin{itemize}
	\item $\Pi^i \subset \Gamma^i_{S}$: a finite subset of policies from which $i$ selects its policy. $\Pi^i$ will be taken to be a fine quantization of $\Gamma^i_{S}$; 
	
	\item UpdateRule$^i \in \PP ( \Pi^i | \Pi^i \times \rr^{\xx \times \uu^i } \times \rr^{\xx} )$ is a stochastic kernel that is used to select a (candidate) successor policy. The distribution over $\Pi^i$ will depend on the current policy of player $i$, its Q-factors (learned or received from an oracle), and its value function estimates (learned or received from an oracle); %

	\item $d^i \geq 0$ is a tolerance for sub-optimality. This is included to account for learning error in the next section; in this section, since there is no learning error, we take $d^i = 0$. 
	\item An experimentation probability $e^i \in ( 0, 1)$: when a player is not $\epsilon$-best-responding, it selects its successor policy according to a mixture distribution, with mixture parts UpdateRule$^i$ and the uniform distribution on $\Pi^i$;
\end{itemize}

\begin{algorithm2e}[h]  
	\SetAlgoLined
	\DontPrintSemicolon
	\SetKw{Receive}{Receive}
	\SetKw{parameters}{Set Parameters}
	\SetKw{initialize}{Initialize}
	\SetKwBlock{For}{for}{end} 
	
	\parameters    \;
	\Indp 
	$\Pi^i \subset \Gamma^i_{S}$: a fine quantization of $\Gamma^i_{S}$ \; 
	UpdateRule$^i \in \PP ( \Pi^i | \Pi^i \times \rr^{\xx \times \uu^i } \times \rr^{\xx} )$: (described above) \;
	$e^i \in (0,1)$: experimentation probability when not $\epsilon$-best-responding \;
	$d^i = 0$: tolerance for sub-optimality \; 
	\Indm 
	\BlankLine
	
	\initialize 	$\pi_0^i \in \Pi^i$: initial policy \;
	
	\BlankLine
	
	\For($k \geq 0$ ($k^{th}$ policy update{)}){ 
		\BlankLine 
		\Receive $Q^{*i}_{\bpi^{-i}_k}$ and $J^i_{\bpi_k}$ \by{given by $J^i_{\bpi_k} ( x ) = J^i ( \bpi_k ,x )$ for all $x \in \xx$.} \; 
		
		\If	{	$J^i_{\bpi_k} ( x ) \leq \min_{a^i \in \uu^i } Q^{*i}_{\bpi^{-i}_k} ( x, a^i ) + \epsilon + d^i \quad \forall x \in \xx$  } 
				{ $\pi^i_{k+1} = \pi^i_k$ }
		\Else(   $\pi^i_k \notin \BR^i_{\epsilon} {(} \bpi^{-i}_{k} {)}$   ){
			$ \pi^i_{k+1} \sim (1- e^i ) {\rm UpdateRule}^i ( \cdot | \pi^i_k, Q^{*i}_{\bpi^{-i}_k} , J^i_{\bpi_k} ) + e^i \uniform ( \Pi^i ) $

			\BlankLine
		}

		Go to $k+1$  \;
	}
	
	\caption{ Randomized Policy Revision for agent $i \in \NN$ (with oracle)} \label{algo:oracle}
\end{algorithm2e}

\begin{lemma} \label{lemma:oracle} 

Let $\GG$ be an $N$-player stochastic game and let $\epsilon \geq 0$. Suppose that $\bPi \subset \bGamma_{S}$ is a finite subset of policies such that $\GG$ has the $\epsilon$-satisficing paths property within $\bPi$. If all players update their policies according to Algorithm~\ref{algo:oracle}, then 
\[
\lim_{ k \to \infty}  \Pr ( \bpi_k \in \eq[\epsilon]_{S} ) = 1.
\]
\end{lemma}

\begin{proof} We have that the stochastic process $\{ \bpi_{k} \}_{k \geq 0}$ is a time homogenous Markov chain on $\bPi$, and that any policy $\bpi^{*} \in \bPi \cap \eq[\epsilon]_{S}$ is an absorbing state for this Markov chain. 

By hypothesis, we have that for each $\tilde{\bpi} \in \bPi$, there exists an $\epsilon$-satisficing path of finite length from $\tilde{\bpi}$ into $\bPi \cap \eq[\epsilon]_{S}$. For each $\tilde{\bpi} \in \bPi$, we let $L_{ \tilde{\bpi}} < \infty$ denote the shortest path of positive probability from $\tilde{ \bpi }$ to some $\epsilon$-equilibrium policy in $\bPi$, and we let $p_{ \tilde{\bpi}} > 0$ be the probability that the Markov chain $\{ \bpi_{k} \}_{k \geq 0}$ follows this path in $L_{\tilde{\bpi}}$ steps conditional on $\bpi_0 = \tilde{\bpi}$. 

We then define $L := \max \{ L_{\tilde{\bpi}} : \tilde{\bpi} \in \bPi \}$ and $p := \min \{ p_{\tilde{\bpi} } : \tilde{\bpi} \in \bPi \}$. We have $L < \infty$ and $p > 0$ by the finiteness of the set $\bPi$. Then, for any $k \geq 0$ we have
\[
\P \left( \bigcap_{j = 1}^M \left\{ \bpi_{k+jL} \notin \eq[\epsilon]_{S} \right\} \middle|  \bpi_k \right) \leq (1-p)^M \to 0, \text{ as } M \to \infty.
\]

\end{proof}

\begin{corollary} \label{corollary:symmetric-oracle} 
Let $\GG$ be an $N$-player symmetric game and let $\epsilon > 0$. Suppose $\bPi \subset \bGamma_{S}$ is a sufficiently fine quantization of $\bGamma_{S}$ such that $\bPi \cap \eq[\epsilon]_{S} \not= \varnothing$ and $\Pi^i = \Pi^j$ for all $i,j \in \NN$. If all players update their policies according to Algorithm~\ref{algo:oracle}, then
\[
\lim_{ k \to \infty}  \Pr ( \bpi_k \in \eq[\epsilon]_{S} ) = 1.
\]
\end{corollary}

\by{We note that, by Corollary~\ref{corollary:quantized-equilibrium}, a policy subset $\bPi$ satisfying the conditions of Corollary~\ref{corollary:symmetric-oracle} can always be found by taking a sufficiently fine quantization of $\bGamma_{S}$.}

\by{
\subsection*{Remarks} From Lemma~\ref{lemma:oracle}, one can see that the existence of $\epsilon$-satisficing paths within a finite subset of policies is, in fact, a sufficient condition for the convergence to $\epsilon$-equilibrium of any $\epsilon$-satisficing process that incorporates random search with positive probability, provided that agents restrict their policies to the finite set in question.

The significance of symmetric games in Corollary~\ref{corollary:symmetric-oracle} is that---because of Corollary~\ref{corollary:symmetric-equality} and the proof of Theorem~\ref{thm:symmetric-paths}---it can be guaranteed that satisficing and quantization are compatible: whenever the quantization $\bPi$ is sufficiently fine and symmetric ($\Pi^i = \Pi^j$ for all $i, j \in \NN$), $\GG$ has the $\epsilon$-satisficing paths property within $\bPi$, by Theorem~\ref{thm:generalized-symmetric-paths}.

When symmetry is not assumed, it is not immediately clear that a given fine quantization will admit $\epsilon$-satificing paths to $\epsilon$-equilibrium. For this reason, we do not provide an analog of Corollary~\ref{corollary:symmetric-oracle} for two-player general sum games, and the learning algorithm and results in the next section are only presented for symmetric $N$-player games.
}
\ \\
Although Algorithm~\ref{algo:oracle} cannot be used \textit{as is} by online independent learners, it does offer insights for the design of independent learners. In particular, both the $\epsilon$-best-responding condition of Line 9 in Algorithm~\ref{algo:oracle} as well as the policy update rule in Line 12 depend on the unobserved joint policy $\bpi_k$ only through the quantities $Q^{*i}_{\bpi^{-i}_{k}}$ and $J^i_{ \bpi_k }$,  \by{which we define by $J^i_{\bpi_k} ( x ) := J^i ( \bpi_k, x ) $ for each $x \in \xx$.}  Crucially, both of these quantities can be estimated during play via learning \textit{without} observing the joint action sequence $\{ \bu^{-i}_t \}_{t \geq 0}$. In the next section, we combine the adaptation mechanism of Algorithm \ref{algo:oracle} with learning to replace the oracle and achieve guarantees on finding $\epsilon$-equilibrium even with independent learners in an online setting.  
 
\subsection{Choice of \by{Parameters} and Update Rule}

We conclude with a brief discussion on the choices of \by{$e^i \in (0,1)$, $\Pi^i$,} UpdateRule$^i \in \PP ( \Pi^i | \Pi^i \times \rr^{\xx \times \uu^i } \times \rr^{\xx} )$ in Algorithm~\ref{algo:oracle} and their effects on convergence to $\eq[\epsilon]_{S}$.

From the proof of Lemma~\ref{lemma:oracle}, one sees that---from the $\epsilon$-satisficing paths property within $\bPi$ alone---random search over $\bPi$ is sufficient to find and stay at an $\epsilon$-equilibrium when using Algorithm~\ref{algo:oracle}. \by{This leads to the requirement that $\bPi$ is selected in such a manner that the game has the $\epsilon$-satisficing paths property within $\bPi$. 

In the context of symmetric stochastic games, this can be done by selecting $\bPi$ to be a symmetric $\xi$-quantization of $\bGamma_{S}$ with $\xi$-sufficiently small so that $\bPi$ contains an $\epsilon$-equilibrium. From a practical point of view, taking the coarsest such quantization is sensible, because it narrows the search space. The required fineness of quantization, as measured by $\xi$, can be determined using the system data (cost functions, discount factors, and the transition kernel) by posing the question as an optimal quantization problem. This is an interesting question that we leave for future research.

Once an appropriate choice of $\bPi$ has been made, we require that $e^i \in ( 0,1)$ for each player $i \in \NN$, so that we are guaranteed that random search can drive play to $\epsilon$-equilibrium within $\bPi$. This is done to avoid \textit{requiring} sophisticated policy update rules that leverage intimate knowledge of the game at hand.} However, \by{selecting $e^i$ too large} may be prohibitively slow, and the choice of UpdateRule$^i$ may result in significant speed-up in driving play to $\epsilon$-equilibrium. 

Thus, the choice of UpdateRule$^i$ represents one area in which algorithm designers can incorporate knowledge of the system being controlled when selecting the particular update rules for the various agents. For instance, inertial best-response dynamics may be appropriate in cooperative settings (as in \cite{AY2017}), while update rules employing gradient ascent/descent may be more appropriate in adversarial settings (as in \cite{daskalakis2020independent}).

\section{Synchronized two timescale learning algorithm} \label{sec:full-algorithm} 

In this section, we present Algorithm~\ref{algo:main}, an independent learning algorithm suitable for online play of symmetric stochastic games under the decentralized information structure of Assumption~\ref{ass:independent-learners}. This algorithm can be interpreted as a noise-perturbed, two-timescale variant of Algorithm~\ref{algo:oracle}, where now action values and state values are estimated rather than obtained from an oracle. 

The algorithm design approach used here builds on a technique presented in \cite{AY2017}, which decouples learning and adaptation: players fix their policies for long intervals of time called exploration phases, during which they update their learning iterates. At the end of an exploration phase, players update their policies using the learned information and then reset their learning iterates ahead of the next exploration phase. This decoupled design is used to mitigate the challenges related to learning in a non-stationary environment, which is among the fundamental difficulties in MARL \cite{hernandez2017survey, zhang2021multi}. At its core, this approach consists of four parts:
\begin{enumerate}
	\item Time is partitioned into intervals called ``exploration phases,'' the $k^{th}$ lasting $T_k \in \nn$ stage games, beginning with the stage game at time $t_k := \sum_{ l = 0}^{k-1} T_l$ and ending after the stage game at $t_k + T_k - 1$;
	
	\item \textit{Within} an exploration phase, agent $i \in \NN$ follows a fixed policy and obtains feedback data on state-action-cost trajectories;
	
	\item \textit{Within} an exploration phase, agent $i$ processes feedback data for policy evaluation, estimation of best-response sets, and estimation of state-action values. This is done without access to the joint action information or knowledge of the joint policy;
	
	\item \textit{Between} the $k^{th}$ and $(k+1)^{th}$ exploration phases, agent $i$ uses the learned information to update its policy from $\pi^i_k$ to $\pi^i_{k+1}$. 	We focus here on $\epsilon$-satisficing update rules.				 \label{adapt-item}
\end{enumerate}

Since the policy of each player is held constant within an exploration phase, this algorithm is not a two-timescale algorithm in the traditional sense (wherein two separate sequences of iterates are updated at every step, with one iterate sequence being updated using an asymptotically larger learning rate than the other sequence), but it is two-timescale in spirit: policy adjustment is done on the slow timescale (indexed by exploration phases) while learning iterates are updated on the fast timescale (indexed by stage games). 

\begin{algorithm2e}[h] 
	\SetAlgoLined
	\DontPrintSemicolon
	\SetKw{Receive}{Receive}
	\SetKw{parameters}{Set Parameters}
	\SetKw{initialize}{Initialize}
	\SetKwBlock{For}{for}{end}
	
	\parameters \;
	\Indp
	$\Pi^i \subset \Gamma^i_{S}$: a fine quantization of $\Gamma^i_{S}$ satisfying Assumption~\ref{ass:quantized-set} \;
	UpdateRule$^i \in \PP ( \Pi^i | \Pi^i \times \rr^{\xx \times \uu^i } \times \rr^{\xx} ) $: described in \S\ref{sec:oracle} \;
	$\mathbb{Q}^i \subset \mathbb{R}^{ \mathbb{X}\times\mathbb{U}^i }$ and $\mathbb{J}^i \subset \rr^{\xx}$: compact sets \; 
	
	$\{ T_k \}_{k \geq 0}$: a sequence in $\nn$ of exploration phase lengths \;
	\vspace{2pt}
	\Indp Put $t_0 = 0$ and $t_{k+1} = t_k + T_k$ for all $k \geq 0.$ \;
	\vspace{2pt}
	\Indm 
	
	$e^i \in (0,1)$: random policy updating probability \;
	$d^i\in(0,\infty)$: tolerance level for sub-optimality \;
	
	\Indm
	\BlankLine
	
	\initialize  $\pi_0^i \in \Pi^i$, $\widehat{Q}_0^i \in \mathbb{Q}^i$, $\widehat{J}^i_0 \in \mathbb{J}^i$ (all arbitrary) \;
	\BlankLine
	\Receive Initial state $x_0 \in \xx$ \\
	
	\For($k \geq 0$ ($k^{th}$ exploration phase{)}){ 
		\For( $t = t_k, t_k +1, \dots, t_{k+1} - 1$  \tcp*[f]{Policy evaluation loop} ){
			Select $u^i_t \sim  \pi^i_k$ 	 \;
			
			\Receive cost $c^i_t := c^i( x_t, u^i_t, \bu^{-i}_t)$, state $x_{t+1}$ \;
			Set $n_t^i := \sum_{ \tau = t_k}^{t} \textbf{1} \{ (x_{\tau}, u^i_{\tau} ) = ( x_t, u^i_t ) \}$ \;
			Set $m_t^i := \sum_{\tau = t_k}^{t} \textbf{1} \{ x_{\tau} = x_t \}$ \;

			\BlankLine
			$\widehat{Q}_{t+1}^i(x_t,u_t^i)  =   \left(1-\frac{1}{n_t^i} \right) \widehat{Q}_t^i(x_t,u_t^i) + \frac{1}{n_t^i} \left[ c^i_t +  \beta^i \min_{v^i} \widehat{Q}_t^i(x_{t+1},v^i) \right]$  \;
			\BlankLine
			 $\widehat{Q}_{t+1}^i (x, a^i) =   \widehat{Q}_t^i(x,a^i )$,   for all $(x,a^i) \not= (x_t,u_t^i)$\;
			\BlankLine
			$\widehat{J}^i_{t+1} ( x_t) = \left(1 - \frac{1}{m^i_t} \right) \widehat{J}^i_t ( x_t ) + \frac{1}{m^i_t} \left[  c^i_t + \beta^i \widehat{J}^i_t ( x_{t+1} )	\right] $ \;
			$\widehat{J}^i_{t+1} ( x ) = \widehat{J}^i_t ( x)$, for all $x \not= x_t$ \;
		}
		\BlankLine
		\tcp*[h]{ Policy Update } 

		\If {$\widehat{J}^i_{t_{k+1}} ( x ) \leq \min_{a^i} \widehat{Q}^i_{t_{k+1}} ( x, a^i )  + \epsilon + d^i$  $\quad \forall x \in \xx$ } {
		$\pi^i_{k+1} \leftarrow \pi^i_k$   } 
		\Else(){ 
			$\pi^i_{k+1} \sim (1-e^i ){\rm UpdateRule}^i ( \cdot | \pi^i_k, \widehat{Q}^i_{t_{k+1}}, \widehat{J}^i_{t_{k+1}} ) + e^i \uniform( \Pi^i )$
		}

		\BlankLine
		
		Reset $\widehat{Q}_{t_{k+1}}^i$ to any $Q^i\in\mathbb{Q}^i$ (e.g., project $\widehat{Q}_{t_{k+1}}^i$ on $\mathbb{Q}^i $) \;
		Reset $\widehat{J}^i_{t_{k+1}}$ to any $J^i \in \mathbb{J}^i$ \; 
	}
	
	\caption{Independent Learning with $\epsilon$-satisficing (for agent $i$)} \label{algo:main}
\end{algorithm2e}

\subsection{Convergence Result for Algorithm~\ref{algo:main}}  We now offer a formal guarantee for the performance of Algorithm~\ref{algo:main} under self-play. Throughout the remainder of this section, we fix the symmetric game $\GG$ and the constant $\epsilon > 0$. We make the following assumptions. 

\begin{assumption} \label{ass:quantized-set}
For a fixed symmetric game $\GG$ and $\epsilon > 0$, the set of policies $\bPi \subset \bGamma_{S}$ is a quantization of $\bGamma_{S}$ with the following properties:
\begin{itemize}
	\item $\Pi^i = \Pi^j$ for each player $i,j \in \NN$;
	\item For any $i \in \NN$ and $ \pi^i \in \Pi^i$, the policy $\pi^i$ is soft; 
	\item The set $\bPi \cap \eq[\epsilon]_{S}$ is non-empty
\end{itemize} 
\end{assumption}

\noindent \by{Such a policy subset $\bPi$ always exists by Corollary~\ref{corollary:quantized-equilibrium} and the discussion of \S\ref{ss:quantization}.}

Next, we introduce a constant $\bar{d} = \bar{d} ( \bPi, \epsilon)$ that depends on the game $\GG$, $\bPi$, and $\epsilon$. We require that the tolerance for sub-optimality $d^i$ is positive to account for noise in the action value and state value estimates; however, $d^i$ cannot be taken too large, otherwise player $i$ may mistakenly suppose it is $\epsilon$-best-responding when using policy that is truly $\epsilon$-suboptimal policy. We let $\bar{d} := \min \left( S \setminus \{ 0 \} \right)$, where $S$ is a finite set defined as 
\[
S := \left\{	\left| \epsilon - \left( J^i ( \bpi ,  x ) - \min_{a^i \in \uu^i} Q^{*i}_{ \bpi^{-i}} (x, a^i) \right) \right|	: i \in \NN, (\pi^i , \bpi^{-i}) \in \bPi, x \in \xx	\right\}. 
\]

\by{
Prior to taking absolute values, elements of the set $S$ can be interpreted as shifted sub-optimality gaps: for each player $i \in \NN$, policy $\bpi = ( \pi^i, \bpi^{-i}) \in \bPi$, and state $x \in \xx$, the quantity $J^i  (\bpi, x ) - \min_{a^i \in \uu^i } Q^{*i}_{\bpi^{-i}} ( x, a^i )  \geq 0$ measure the degree to which player $i$'s policy $\pi^i$ is sub-optimal against $\bpi^{-i}$ in state $x$. Thus, measuring 
\[
\left| \epsilon - \left( J^i  (\bpi, x ) - \min_{a^i \in \uu^i } Q^{*i}_{\bpi^{-i}} ( x, a^i )  \right) \right|  
\]
captures how close player $i$ is to \emph{exactly} $\epsilon$-best-responding in state $x$, and $\bar{d}$ is the minimum non-zero discrepancy from exact $\epsilon$-best-responding. 
}

\begin{assumption} \label{ass:rho-delta}
For all $i \in \NN$,  $d^i \in (0, \bar{d} \,  )$.
\end{assumption}

We also make the following standard assumption on the state transition kernel $P$, which is necessary to ensure that no proper subset of states is absorbing.
\begin{assumption} \label{ass:states}
For any two state $x, x' \in \xx$, there exists $H = H(x,x') \in \nn$ and sequences of states $\{ s_k \}_{k = 0}^{H+1}$ and joint actions \emph{$\{ \ba_k \}_{k = 0}^{H}$} such that $s_0 = x$, $s_{H+1} = x'$, and 
\[
	\text{\emph{$\prod_{j = 0}^H P ( s_{j+1} | s_j, \ba_j ) > 0,$}}
\]
\end{assumption}

\by{
Assumption~\ref{ass:states} requires that there is a non-zero probability of transitioning from any initial state to any other state in a finite number of steps, provided the players select an appropriate sequence of joint actions. Both the number of steps and the sequence of joint actions in question are allowed to depend on the pair of states. As such, this assumption is quite weak, and is commonly made elsewhere in the literature. (See, for instance, \cite[Assumption 2-i]{sayin2021decentralized}, for an equivalent assumption on the transition kernel.) This assumption is necessary to ensure that no proper subset of states is absorbing.
}

\begin{theorem} \label{thm:main}
Let $\GG$ be a symmetric game, $\epsilon> 0$. Suppose all players use Algorithm~\ref{algo:main} to play the game $\GG$, and that Assumptions~\ref{ass:quantized-set}, \ref{ass:rho-delta}, and \ref{ass:states} hold. Then, for any $\psi > 0$, there exists $\tilde{T} = \tilde{T} ( \psi, \epsilon, \bPi, \{ d^i \}_{i \in \NN} )$ such that if $T_k \geq \tilde{T}$ for all exploration phase indices $k \geq 0$, then
\[
	\Pr \left( \bpi_k \in \bPi \cap \eq[\epsilon]_{S} \right) \geq 1 - \psi , \quad \text{for all sufficiently large $k$.}
\]
\end{theorem}

\noindent The proof of Theorem~\ref{thm:main} is given in Appendix~\ref{appendix:pf-main}. \\


\by{
\subsection*{Remarks} 
We have chosen to present Algorithm~\ref{algo:main} and Theorem~\ref{thm:main} in terms of symmetric games and the quantized set of policies $\bPi$ to take advantage of the structural properties established in earlier sections. It is worth noting that the algorithm and convergence guarantees above can be applied more generally: in analogy to Lemma~\ref{lemma:oracle} and the remarks following Corollary~\ref{corollary:symmetric-oracle}, Algorithm~\ref{algo:main} can be used to drive play to $\epsilon$-equilibrium in non-symmetric games, provided the game has the $\epsilon$-satisficing paths property within the finite set $\bPi$ and the policies in $\bPi$ are soft. These conditions can also be shown to hold for (non-symmetric) stochastic teams and exact potential games, among other classes of games. 
}

\section{Simulations} \label{sec:simulation} \by{We now present the results of a simulation study of Algorithm~\ref{algo:main} applied to a symmetric stochastic, described below. }

\begin{figure}[h]
\centering
\begin{subfigure}{.5\textwidth}
  \centering
\begin{game}{2}{2}[]
			&$a_0$      							&$a_1$					\\
$a_0$ 		&5,5	  			&0,0 										\\
$a_1$		&0,0				&5,5				
\end{game}
  \caption{State $s_0$: a coordination game}
  \label{fig:sub1}
\end{subfigure}%
\begin{subfigure}{.5\textwidth}
  \centering
\begin{game}{2}{2}[]
			&$a_0$      							&$a_1$					\\
$a_0$ 		&0.75, 0.75 	  					&1.5,0.5		\\
$a_1$		&0.5,1.5							&1,1	
\end{game}
  \caption{State $s_1$: prisoner's dilemma}
  \label{fig:sub2}
\end{subfigure}
\caption{The stage games for a two-state stochastic game. Player 1 (2) picks a row (column), and its reward, to be maximized, is the 1st (2nd) entry of the chosen cell.}
\label{fig:two-player-game}
\end{figure}

\by{
Here, the player set is $\NN = \{ 1, 2\}$, the state space $\xx = \{ s_0, s_1 \}$, and the action sets $\uu^1 = \uu^2 = \{ a_0, a_1 \}$.  The cost functions are taken to be the negatives of the stage reward functions described in Figure~\ref{fig:two-player-game}. The discount factor is $\beta^i = 0.8$ for each $i \in \NN$, and the transition kernel $P$ is fully described by the following equations.

\begin{align*}
\P \left( x_{t+1} = s_0 \middle| x_t = s_0, u^1_t = b^1, u^2_t = b^2 \right) &= \begin{cases}
																0.8, &\text{if } b^1 = b^2 \\
																0.2, &\text{otherwise.}
															\end{cases}
														\\
\P \left( x_{t+1} = s_0 \middle| x_t = s_1, u^1_t = b^1, u^2_t = b^2  \right) &= \begin{cases}
																0.9, &\text{if } b^1 = b^2 = a_1 \\
																0.25, &\text{otherwise.}
															\end{cases}											
\end{align*}

The game described above involves two states with rather different strategic qualities: in state $s_0$, the players are incentivized to coordinate actions, so as to receive a high reward and remain in the high-value state $s_0$. By contrast, in state $s_1$, players face the prisoner's dilemma stage game, with the added consideration that successfully cooperating (playing action $a_1$) drives play back to the high-value state $s_0$ with high probability, but cooperating while the other player defects (plays $a_0$) results in a lower immediate reward and does not increase the likelihood of transitioning to state $s_0$. 
}

\by{
\subsection*{Parameter Choices}

For each $i \in \NN$, we chose the quantized policy sets $\Pi^i$ as follows. First, we define $\tilde{\Pi}^i \subset \Gamma^i_{S}$ as $\tilde{\Pi}^i := \{ \pi^i \in \Gamma^i_{S} : 10^2 \pi^i ( a^i | s ) \in \zz , \forall s \in \xx, a^i \in \uu^i \}$. That is, policies in $\tilde{\Pi}^i $ can be described using two digit precision after the decimal point in each component probability distribution. We then define a function ${\rm Soft}: \Gamma^i_{S} \to \Gamma^i_{S}$ by
\[
{\rm Soft}(\pi^i) := 	\begin{cases}
					\pi^i, &\text{if } \pi^i \text{ is 0.025-soft}, \\
					0.9 \pi^i + 0.1 \cdot \pi^i_{\rm unif} &\text{otherwise},
				\end{cases}
\]
where $ \pi^i_{\rm unif} $ is the policy that selects actions uniformly at random in each state. Finally, we put $\Pi^i = {\rm Soft} ( \tilde{\Pi}^i )$. 

We put $\epsilon = 2$. Although this may appear to be a rather large choice of $\epsilon$, the relatively large discount factor of $\beta^i = 0.8$ leads to aggregate long-run rewards that are an order of magnitude larger. Empirically, we found that performance within $\epsilon = 2$ of a 0-best-response entailed achieving over 80\% of one's optimal return in state $s_1$ and over 85\% of one's optimal return in state $s_0$.

We ran 250 trials of Algorithm~\ref{algo:main}. Each trial consisted of 50 exploration phases of length 50,000, for a total of 2.5 million stage games per trial. For our choice of policy update rule, we set UpdateRule$^i$ to be a gradient ascent-type policy update, in which the player increments its policy toward a 0-best-response using a small step size. Our empirical results, in which we observe the frequency of $\epsilon$-equilibrium rising to over 90\% of trials,  are summarized in Figure~\ref{fig:graph} and Table~\ref{table:some-results}. 

\begin{figure} 
\includegraphics[width = 0.85\textwidth]{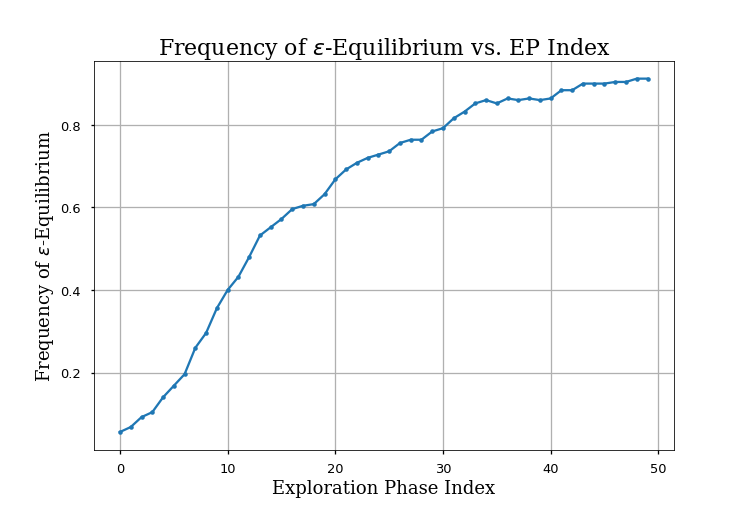}   
\centering
\caption[]{Frequency of $\{ \bpi_k \in \eq[\epsilon]_{S} \cap \bPi \}$, averaged over 250 trials.}
\label{fig:graph}
\end{figure}

\begin{table}
\renewcommand{\arraystretch}{1.5}
\begin{tabular}{ c | c | c | c | c | c | c }
EP Index						& 1 		& 10 			& 25		&40 		& 45 		&  50 \\
\hline
$\frac{1}{250} \sum_{k=1}^{250} \textbf{1} \{ \bpi_k \in \eq[\epsilon] \}$	&0.056	&0.356		& 0.728	&0.86	&0.9		& 0.912
\end{tabular}
\centering
\caption{Selected frequencies of $\epsilon$-equilibrium at various exploration phase indices}   \label{table:some-results}
\end{table}
}

\section{Discussion}  \label{sec:discussion}

Algorithm~\ref{algo:main} is in the tradition of Foster and Young's regret testing algorithm \cite{foster2006regret}. Their seminal algorithm was developed for stateless repeated games and came with convergence guarantees for general (stateless, finite) two-player games. A variant of the regret testing algorithm was studied in \cite{germano2007global}, where it was shown that the variant algorithm converges to equilibrium in any \textit{generic} $N$-player stateless finite game.\footnote{We note that the class of generic games is a proper subclass of all games. As such, this guarantee may not hold for finite stateless games in general.} Other papers in the regret testing tradition include \cite{Marden2009payoff} and \cite{AY2017}, which studied algorithms in weakly acyclic games.

Unlike earlier contributions in the regret testing tradition, the line of proof used here relies on the novel structure of $\epsilon$-satisficing paths. Similarly, while other rigorous contributions to the study of independent learners in stochastic games have relied heavily on various structural assumptions (most notably focusing on two-player zero-sum games, $N$-player teams, and potential games), it appears that the rich structure of $\epsilon$-satisficing paths has been under-exploited. To our knowledge, Algorithm~\ref{algo:main} is the first independent learner for general symmetric stochastic games that comes with formal guarantees of convergence to $\epsilon$-equilibrium. 

Like the previous algorithms in the regret testing tradition, one shortcoming of Algorithm~\ref{algo:main} is that players update their policies synchronously at the end of each synchronized exploration phase. This may be justifiable in certain settings where time is naturally partitioned (e.g. competitions with a natural off-season; games where certain stage game decision must be made quickly while others can be made slowly) or in cooperative applications; however, in other settings it may be inappropriate. Empirically, it appears that perfect synchrony is not needed, and \cite{Marden2009payoff} offers a possible approach to formalizing this, but as of yet, no proof has been given. This issue can potentially be addressed by using a conventional two-timescale algorithm, in which agents update their policies after every stage game with a slow learning rate while they update their learning iterates using a fast learning rate.


\subsection{Future Work}

It is natural to ask whether all $N$-player stochastic games have the $\epsilon$-satisficing paths property for any $\epsilon \geq 0$. It is clear that the proof of Theorem~\ref{thm:symmetric-paths} is not amenable to generalization, since in general games players will not have matching policy sets and therefore cannot imitate one another's policies to grow the size of the cohort. Generalizing the proof of Theorem~\ref{thm:two-player-paths} is perhaps more promising, but the case-by-case breakdown used there becomes significantly complicated when $N > 2$. We leave this question open for future research.

The study of learning in symmetric games is applicable to several other technical fields. One such field is the emerging area of mean field game theory, which has been used to model various large-scale decentralized decision-making systems, such as 
traffic networks \cite{chevalier2015micro}. 
Mean-field games (see e.g., \cite{huang2006large, huang2007large, lasry2007mean,carmona2016mean}) 
can be viewed as limit models of symmetric games with finitely many agents and weakly coupled interactions. A number of works have investigated the deep connection between finite symmetric games and mean-field games with a continuum of players, e.g. \cite{fischer2017connection, sanjari2018optimal, sanjari2019optimal, sanjari2020optimality}. Therefore, results on equilibrium, learning, and dynamics in symmetric games are consequential for developing a theory of learning in large-scale, decentralized systems. \by{Some of the ideas presented in this paper have been applied to $N$-player mean field games in \cite{yongacoglu2022independent}.}

\by{Pertaining specifically to the algorithms presented here, an important question for future work is to determine how one should select the quantizer determining the finite policy sets $\{ \Pi^i \}_{i \in \NN}$ so as to achieve optimal performance with the algorithm. When selecting a $\xi$-quantizer, one must balance selecting sufficiently small quantizer bin radius $\xi$, which is needed to guarantee that some $\epsilon$-equilibrium exists in the resulting quantization, with ensuring that the cardinality of the quantized set is relatively small, to avoid searching a needlessly large search space and slowing down convergence to $\epsilon$-equilibrium.}

\by{In the same vein, another important question for future work involves generalizing the learning procedure. In this paper, we have focused on (standard) tabular Q-learning and state value estimation using linear learning rates to produce each player's estimate of whether it is $\epsilon$-best-responding. This choice was made for simplicity of exposition and in order to make rigorous claims about convergence behaviour. Since Q-learning can be quite slow to converge (see, for instance, \cite{szepesvari1997asymptotic,even2003learning}), it would be desirable to study variants of our algorithm that employ other learning algorithms, such as Speedy Q-learning \cite{azar2011speedy}, Zap Q-learning \cite{devraj2017zap}, or some form of function approximation in an attempt to shorten exploration phase lengths.}

\by{
\subsection{On Terminology and Other Related Work}

We wish to conclude our discussion by situating our work in the broader literature on multi-agent reinforcement learning.

In using the terminology of ``independent learners'' to describe learners in stochastic games with full state observability at each agent but no action-sharing between agents, we follow the terminological conventions of \cite{Claus1998} and \cite{matignon2009coordination}. Given that the field of MARL is relatively young, this convention, like many others, is not uniform. Some authors use the term ``independent learner'' to describe self-interested learners who employ best-response or policy gradient-type algorithms in stochastic games, and this terminology does not convey any assumptions about action-sharing. For examples of work that uses this language, see the recent survey by \cite{ozdaglar2021independent} and the references therein. Using the language of \cite{ozdaglar2021independent} , the object of our study is model-free MARL in the minimal information setting.

Many studies in MARL, including those presented here, are concerned with asymptotic convergence guarantees. A separate line of results is concerned with giving non-asymptotic guarantees in stochastic games, such as regret bounds for the performance of a particular agent in a multi-agent system. We selectively cite \cite{wei2017online,bai2020near} and \cite{liu2021sharp} as results in this second line, and we refer the reader to \cite[Section 5]{ozdaglar2021independent} for an excellent review of recent work in this area.

A third line of research in MARL is concerned with the complexity of computing various equilibrium concepts in stochastic games, including the recent work of \cite{daskalakis2022complexity}, which establishes the PPAD-hardness of computing (stationary Markov) $\epsilon$-coarse correlated equilibria in stochastic games. We wish to point out that there is no inherent conflict in the negative results of \cite{daskalakis2022complexity} and the positive results presented here: firstly, we study randomized algorithms and give high probability guarantees, and, secondly, we do not offer complexity analysis for our algorithms. Furthermore, although we prove the existence of short $\epsilon$-satisficing paths to $\epsilon$-equilibrium in two-player general sum games and in $N$-player symmetric games, we do not study the complexity of computing such a path. 
}

\section{Conclusions} \label{sec:conclusion}

In this paper, we introduced the $\epsilon$-satisficing paths property, a useful structural property pertaining to policy revision dynamics in stochastic games. We proved that two important classes of games, namely $N$-player symmetric games and two-player general games, both have this property. In the case of $N$-player symmetric games, we have exploited this structure to design an independent learner and showed that this algorithm drives play to near equilibrium with arbitrarily high probability under self-play. This is the first result of its kind for this class of games, and we believe that a similar design approach can be used to establish analogous results in other classes of games admitting the $\epsilon$-satisficing paths property for $\epsilon > 0 $.

\newpage
\appendix
\section{Proof of Theorem~\ref{thm:main}} \label{appendix:pf-main} In this section, we prove Theorem~\ref{thm:main}. To study the evolution of the policy process $\{ \bpi_k \}_{k \geq 0}$ obtained by Algorithm~\ref{algo:main}, we first study the convergence behaviour of the learning iterates $\{ \widehat{Q}^i_t \}_{t \geq 0}$ and $\{ \widehat{J}^i_t \}_{t \geq 0}$ and then argue that properly selected parameters result in policy updates similar to those obtained by the oracle process studied in Lemma~\ref{lemma:oracle}.

We note that when agent $i$ uses Algorithm~\ref{algo:main} to select its actions, it is using a particular randomized, non-stationary policy. When all agents employ Algorithm~\ref{algo:main}, we use $\P$ (with no policy index in the superscript) to denote the probability measure on trajectories of states and joint actions. In contrast, for all other joint policies $\bpi \in \bGamma$, we use $\P^{\bpi}$ to denote the associated probability measure on trajectories of states and joint actions. This distinction is made to facilitate comparing and analyzing various stochastic learning iterates.

\subsection{Convergence Behaviour of Learning Iterates}
We begin by studying the convergence of player learning iterates $\{ \widehat{Q}^i_t \}_{t \geq 0}$ and $\{ \widehat{J}^i_t \}_{t \geq 0}$. Since the policy updates are informed by the iterate sequences only at the end of each exploration phase, we give special attention to the iterate values at the sample times $\{ t_{k+1} \}_{k \geq 0}$. That is, we are interested in the sequences $\{ \widehat{Q}^i_{t_{k+1}} \}_{k \geq 0}$ and $\{ \widehat{J}^i_{t_{k+1}} \}_{k \geq 0}$. 

In the interest of comparing the iterate sequences $\{ \widehat{Q}^i_t \}_{t}$ and $\{ \widehat{J}^i_t \}_{t}$ with easily analyzed iterate sequences, we now introduce two related stochastic processes, $\{ \bar{Q}^i_t \}_{t \geq 0}$ and $\{ \bar{J}^i_t \}_{t \geq 0}$. The latter sequences are obtained using the state-action-cost trajectories using Algorithm~\ref{algo:Q-and-J}, below.

\begin{algorithm2e}[h]  
	\SetAlgoLined
	\DontPrintSemicolon
	\SetKw{inputs}{Inputs}
	\SetKwBlock{For}{for}{end} 
	
	\inputs    \;
	\Indp 
	$\bar{Q}^i_0 \in \mathbb{Q}^i \subset \rr^{\xx \times \uu^i}$, where $\mathbb{Q}^i$ is compact \; 
	$\bar{J}^i_0 \in \mathbb{J}	^i \subset \rr^{\xx}$, where $\mathbb{J}^i$ is compact \; 
	Trajectory $\{ ( x_t, u^i_t, c^i_t ) \}_{t \geq 0}$, where $c^i_t = c^i ( x_t, u^i_t, \bu^{-i}_t)$ for all $t \geq 0$\; 
	\Indm 
	\BlankLine
		
	\For($t \geq 0$ ){ 
		$n^i_t := \sum_{\tau = 0}^{t} \textbf{1} \{ ( x_{\tau}, u^i_{\tau} ) =( x_t, u^i_t ) \}  $   \;
		$m^i_t := \sum_{\tau = 0}^t  \textbf{1 } \{ x_\tau = x_t \} $ \; 
		\BlankLine 
		$\bar{Q}^i_{t+1} (x_t, u^i_t ) = (1 - 1/ n^i_t ) \bar{Q}^i_t ( x_t, u^i_t ) + (1/n^i_t ) \left[ c^i_t + \beta \min_{ \tilde{a}^i \in \uu^i } \bar{Q}^i_t ( x_{t+1} , \tilde{a}^i ) \right] $
		\BlankLine
		$\bar{Q}^i_{t+1} ( s,a^i) = \bar{Q}^i_{t} (s,a^i)$ for all $(s,a^i) \not= ( x_t, u^i_t )$. 

		\BlankLine
		$\bar{J}^i_{t+1} ( x_t ) = (1 - 1 / m^i_t ) \bar{J}^i_t ( x_t ) + (1 / m^i_t ) \left[ c^i_t + \beta \bar{J}^i_t ( x_{t+1} ) \right].$
		\BlankLine
		$\bar{J}^i_{t+1} ( s ) = \bar{J}^i_t (s)$ for all $s \not= x_t$. 
		
	}
	
	\caption{Q- and J-factor Updating Without Resetting} \label{algo:Q-and-J}
\end{algorithm2e}

The sequences $\{ \widehat{Q}^i_t \}_{t \geq 0}$ and $\{ \bar{Q}^i_t \}_{t \geq 0}$ are related through the Q-factor update. There are, however, two major differences. First, Algorithm~\ref{algo:main} instructs player $i$ to reset its counters at the end of the $k^{th}$ exploration phase (i.e. after the update at time $t_{k+1}$, before the update at time $t_{k+1} + 1$), meaning the step sizes differ for the two iterate sequences $\{ \bar{Q}^i_t \}_{ t \geq 0} $ and $\{ \widehat{Q}^i_t \}_{ t \geq 0}$ even when the state-action-cost trajectories are identical. Second, Algorithm~\ref{algo:main} instructs player $i$ to reset its Q-factors at the end of the $k^{th}$ exploration phase, while Algorithm~\ref{algo:Q-and-J} does not involve any resetting. 

Consequently, one sees that the process $\{ \widehat{Q}^i_t \}_{t \geq 0}$ depends on the initial condition $\widehat{Q}^i_0$, the state-action trajectory, and the choice of reset values after each exploration phase. In contrast, the process $\{ \bar{Q}^i_t \}_{t \geq 0}$ depends only on the initial value $\bar{Q}^i_0$ and the state-action trajectory. Analogous remarks hold relating $\{ \widehat{J}^i_t \}_{ t \geq 0}$ and $\{ \bar{J}^i_t \}_{ t \geq 0}$. 

Recall that the $k^{th}$ exploration phase begins with the stage game at time $t_k$ and ends before the stage game at time $t_{k+1} = t_k + T_k$. During the $k^{th}$ exploration phase, the sequences $\{ \widehat{Q}^i_t \}_{t = t_k }^{ t_k + T_k }$ and $\{ \widehat{J}^i_t \}_{t = t_k}^{t_k + T_k }$ depend only on the initial conditions $\widehat{Q}^i_{t_k}$, $\widehat{J}^i_{t_k}$, and the state-action trajectory $x_{t_k}, \bu_{t_k}, \cdots, \bu_{t_k + T_k - 1 }, x_{t_k + T_k}$. This leads to the following useful observation: for any $( s_0, s_1, \cdots, s_{T_k} ) \in \xx^{ T_k + 1 }$ and $( \ba_0, \cdots, \ba_{T_k - 1 } ) \in \bU^{ T_k }$, we have that 
\begin{align*}
\P &\left( \{ x_{t_k + T_k} = s_{T_k} \} \bigcap_{j = 0}^{T_k - 1} \{ x_{t_k + j } = s_{j} , \bu_{t_k + j } = \ba_j  \}  \middle| x_{t_k} = x , \bpi_k = \bpi \right)  \\
= \P^{\bpi} &\left( \{ x_{T_k} = s_{T_k} \} \bigcap_{j = 0}^{T_k-1} \{ x_j = s_j , \bu_j = \ba_j \} \middle| x_0 = x \right) .
\end{align*}

In words, once players following Algorithm~\ref{algo:main} select a policy $\bpi$ for the $k^{th}$ exploration phase, then the conditional probabilities of the trajectories restricted to time indices in that exploration phase can be described by $\P^{\bpi}$, with the indices of the events suitably shifted to start at time 0. This leads to a series of useful lemmas, which we include below for completeness.

\by{In the lemmas below, for notational convenience, we define $J^i_{\bpi} : \xx \to \rr$ as $J^i_{\bpi} ( s ) = J^i ( \bpi, s )$ for all $s \in \xx$ and any $i \in \NN$, $\bpi \in \bGamma_{S}$.}

\begin{lemma} \label{lemma:relating-the-sequences}
Let $\bpi \in \bPi \subset \bGamma_{S}$ be some joint policy and let $i \in \NN$. For any initial conditions $Q^i \in \mathbb{Q}^i$, $J^i \in \mathbb{J}^i$, and state $x \in \xx$, we have the following:
\[
\P \left( \widehat{Q}^i_{t_k + T_k} \in \cdot \middle| \widehat{Q}^i_{t_k} = Q^i, \bpi_k = \bpi, x_{t_k} = x \right) = \P^{\bpi} \left( \bar{Q}^i_{T_k} \in \cdot \middle| \bar{Q}^i_0 = Q^i , x_0 = x \right), 
\]
and 
\[
\P \left( \widehat{J}^i_{t_k + T_k} \in \cdot \middle| \widehat{J}^i_{t_k} = J^i, \bpi_k = \bpi, x_{t_k} = x \right) = \P^{\bpi} \left( \bar{J}^i_{T_k} \in \cdot \middle| \bar{J}^i_0 = J^i, x_0 = x \right). 
\]
\end{lemma}

\begin{lemma}  \label{lemma:uniform-conditional-convergence} 
For any joint policy $\bpi \in \bPi$, player $i \in \NN$, we have the following:
\begin{itemize}
	\item[(i)] $\P^{\bpi} \left( \lim_{t \to \infty} \bar{Q}^i_t = Q^{*i}_{ \bpi^{-i}} \middle| \bar{Q}^i_0 = Q^i, x_0 = x \right) = 1$, for any $x \in \xx$, $Q^i \in \rr^{\xx \times \uu^i}$; 
	
	\item[(ii)] $\P^{\bpi} \left( \lim_{t \to \infty} \bar{J}^i_t = J^i_{\bpi} \middle| \bar{J}^i_0 = J^i, x_0 = x \right) = 1$, for any $x \in \xx$, $J^i \in \rr^{\xx}$; 
		
	\item[(iii)] If $\mathbb{Q}^i \subset \rr^{\xx \times \uu^i}$ and $\mathbb{J}^i \subset \rr^{\xx}$ are compact sets, then items (i) and (ii) hold uniformly in their initial conditions. That is, for any $\xi > 0$, there exists $T = T ( i, \xi, \bpi ) \in \nn $ such that
		\[
		\P^{\bpi} \left( \sup_{t \geq T } \left\| \bar{Q}^i_t - Q^{*i}_{\bpi^{-i}} \right\|_{\infty} < \xi \middle| \bar{Q}^i_0 = Q^i, x_0 = x \right) \geq 1 - \xi,
		\]
		and 
		\[
		\P^{\bpi} \left( \sup_{t \geq T } \left\| \bar{J}^i_t - J^i_{\bpi} \right\|_{\infty} < \xi \middle| \bar{J}^i_0 = J^i, x_0 = x \right) \geq 1 - \xi,
		\]
		for any $Q^i \in \mathbb{Q}^i, J^i \in \mathbb{J}^i, x \in \xx$. 
\end{itemize}
\end{lemma}

\begin{proof}
Items (i) and (ii) are proved using standard stochastic approximation arguments, e.g. those in \cite{tsitsiklis1994asynchronous}. We note that each state-joint action pair is visited infinitely often $\P^{\bpi}$-almost surely. This is obtained by Assumptions~\ref{ass:quantized-set} and \ref{ass:states}: the former which guarantees that $\pi^j$ is soft for every player $j \in \NN$, and the latter guarantees that all states are mutually accessible.

The uniformity claim, in item (iii), was proven for Q-factors alone in \cite[Lemma A1]{AY2017}. The same line of argument can be used for the iterates $\{ \bar{J}^i_t \}_{t \geq 0}$. 
\end{proof}

Combining Lemmas~\ref{lemma:relating-the-sequences} and \ref{lemma:uniform-conditional-convergence}, we get the following result on conditional probabilities

\begin{lemma}  \label{lemma:all-learn-correctly}
Let $k, \ell \in \zz_{\geq 0}$ such that $\ell \geq k$. Let $\mathcal{F}_k$ denote the $\sigma$-algebra generated by the random variables $
 \bpi_k, \{ \widehat{Q}^i_{t_k} , \widehat{J}^i_{t_k} \}_{i \in \NN},$ and $x_{t_k}$. For any $\xi > 0$, there exists $T = T(\xi) < \infty$ such that if $T_{\ell} \geq T$, then $\P$-almost surely, we have 
\[
\P \left( \bigcap_{i \in \NN} \left\{ \left\| \widehat{Q}^i_{t_{ \ell +1}} - Q^{*i}_{\bpi^{-i}_{\ell} } \right\|_{\infty} < \xi \right\} \cap \left\{  \left\| \widehat{J}^i_{t_{\ell+1}} - J^i_{\bpi_{\ell}} \right\|_{\infty} < \xi \right\} \middle| \mathcal{F}_k \right) \geq 1  - \xi. 
\]
\end{lemma}

\subsection{Proof Details}

Let $\Xi := \frac 1 2 \min_{j \in \NN} \{ d^j, \bar{d} - d^j \}$, and for $k \geq 0$, let $E_k$ denote the event that each player $i$ learned its Q- and J-factors to within $\Xi$ of their fixed points during the $k^{th}$ exploration phase. That is, for any $k \geq 0$,
\[
E_k := \bigcap_{i \in \NN} \left\{ \left\| \widehat{Q}^i_{t_{ k +1}} - Q^{*i}_{\bpi^{-i}_{k} } \right\|_{\infty} < \Xi \right\} \cap \left\{  \left\| \widehat{J}^i_{t_{k+1}} - J^i_{\bpi_{k}} \right\|_{\infty} < \Xi \right\} .
\]
For any $\ell \geq 0$, let $E_{k: k+ \ell} := E_k \cap E_{k+1} \cap \cdots \cap E_{k+\ell}$. For any $k \geq 0$, we also define $G_k := \{ \bpi_k \in \bPi \cap \eq[\epsilon]_{S} \}$ to be the event that the policy of the $k^{th}$ exploration phase is an $\epsilon$-equilibrium. 

By our definition of $\bar{d}$, assumption that $d^i \in (0, \bar{d})$ for each $i \in \NN$ (Assumption~\ref{ass:rho-delta}), and our choice of $\Xi$, we have that, given $E_k$, each player $i \in \NN$ correctly ascertains whether $\pi^i_k \in \BR^i_{\epsilon} ( \bpi^{-i}_k )$ by verifying whether $\widehat{J}^i_{t_{k+1}} ( x ) \leq \min_{a^i } \widehat{Q}^i_{t_{k+1}} + \epsilon + d^i$ for each $x \in \xx$. From this, it follows that
\begin{equation} \label{eq:stay} 
\P \left( G_{k+ \ell} \middle| G_k \cap E_{k : k + \ell } \right) = 1, \: \forall \ell \geq 0. 
\end{equation}

Recall the quantity $L := \max \{ L_{\bpi_0 } : \bpi_0 \in \bPi \}$, where for each $\bpi_0 \in \bPi$, $L_{\bpi_0}$ is defined as the shortest $\epsilon$-satisficing path within $\bPi$ beginning at $\bpi_0$ and ending in $\bPi \cap \eq[\epsilon]_{S}$. For any $\bpi \in \bPi$, there exists an $\epsilon$-satisficing path from $\bpi$ into $\bPi \cap \eq[\epsilon]_{S}$, and if this path has length less than $L$, it may be extended to have length $L$ by repeating its final term. Thus, for any $\bpi \in \bPi$, we have the following inequality:
\begin{equation} \label{eq:go}
\P \left( G_{k+L} \middle| \{ \bpi_k = \bpi \} \cap E_{k : k+L} \right) \geq p_{\min} > 0,
\end{equation}
where $p_{\min} := \prod_{i \in \NN} \left(	\frac{ e^i }{ | \Pi^i | } \right)^{L}$. The quantity $p_{\min}$ is a loose lower bound obtained as follows: starting at $\bpi_k = \bpi$, at each step, suppose that every unsatisfied player chooses to experiment (with probability $e^i$ at each step) and selects the policy that follows the specified $\epsilon$-satisficing path by (with probability $1/ | \Pi^i |$ at each step). There are at most $L$ such steps, therefore the probability of following the specified path by random experimentation alone is no less than $p_{\min}$.

Fix $u^{*} \in (0,1)$ such that $\frac{ u^* p_{\min} }{  1 - u^* + u^* p_{\min} } > 1 - \psi/2$. As a consequence of Lemma~\ref{lemma:all-learn-correctly}, there exists $\tilde{T} < \infty$ such that if $T_l \geq \tilde{T}$ for every $l \geq 0$, then we have $\P ( E_{k : k+L} | \bpi_k = \bpi ) \geq u^*$ for all $k \geq 0$ and any $\bpi \in \bPi$. Thus, for any $k \geq 0$, we have that 
\[
\P ( E_{k: k+L} | G_k ) \geq u^* \text{ and }  \P ( E_{k:k+L} | G_k^c ) \geq u^* .
\]

For any $k \geq 0$, we can lower bound $\P ( G_{k+L} )$ by conditioning on $G_k$ and its complement:
\[
\P ( G_{ k + L } ) = \P ( G_{k+L} | G_k ) \P ( G_k ) + \P ( G_{k+L}  | G_k^c ) ( 1 - \P (G_k ) ).
\]

We lower bound the constituent terms above by conditioning again with $E_{k : k+L}$ and invoking \eqref{eq:stay} and \eqref{eq:go} to get
\[
\P ( G_{k+L} ) \geq 1 \cdot  \P ( E_{k : k+L} | G_k ) \cdot \P (G_k )+ p_{\min} \cdot \P ( E_{k:k+L} | G_k^c) \cdot (1 - \P (G_k )).
\]

Assuming $T_l \geq \tilde{T}$ for all $l \geq 0$, this gives 
\[
\P ( G_{k+L} ) \geq   u^* \cdot \P (G_k )+ p_{\min} \cdot u^* \cdot (1 - \P (G_k )), \quad \forall k \geq 0.
\]

For each $k \in \{0 , 1, \dots, L - 1 \}$, define $y^{(k)}_0 := \P ( G_k )$, and for $m \geq 0$, define $y^{(k)}_{m+1} = u^* y^{(k)}_m + u^* p_{\min} \left( 1 - y^{(k)}_m  \right)$. One can use an inductive argument to show that 
\begin{equation} \label{eq:lower-bound-on-equilibrium} 
\P ( G_{k + mL } ) \geq y^{(k)}_{m} \quad \forall m \geq 0. 
\end{equation}

Observe that $y^{(k)}_{m+1}$ can be written as 
\[
y^{(k)}_{m+1} = \left( u^* - u^* p_{\min} \right)^{m+1} y^{(k)}_0 + u^* p_{\min} \sum_{j = 0}^{m} \left( u^* - u^* p_{\min} \right)^j .
\] 

Since $0 < u^* - u^* p_{\min} < 1$, we have that $\lim_{m \to \infty} y^{(k)}_m = \frac{ u^* p_{\min} }{ 1 - u^* + u^* p_{\min} } > 1 - \frac{\psi}{2}$. Thus, by \eqref{eq:lower-bound-on-equilibrium}, we have that $\P ( G_{k + mL} ) > 1 - \psi/2$ holds for all sufficiently large $m$, which proves the result.

\ \\

\section{Proof of Lemmas~\ref{lemma:continuous-cost}--\ref{lemma:continuous-max-state}}   \label{appendix:proofs-of-continuity-lemmas}    

\begin{lemma} \label{lemma:continuous-probabilities}
Fix $i \in \NN$ and $T \in \nn$. For any $(s,a^i) \in \xx \times \uu^i$ and \emph{$( \tilde{s}_k, \tilde{\ba}_k )_{k = 0}^T \in  \left( \xx \times \bU \right)^{T+1}$}, 
the mapping $F : \bGamma_{S} \to \rr$ given by 
\[
F( \bpi ) := \P^{\bpi} \left[ \bigcap_{k= 0}^T \{ x_k = \tilde{s}_k , \emph{$ \bu_k = \tilde{\ba}_k $} \} \middle| x_0 = s, u^i_0 = a^i \right], \quad \forall \bpi \in \bGamma_{S}
\]
is continuous.
\end{lemma}

\begin{proof}
We write the specified probability as 
\begin{align*}
&\P^{\bpi} \left[ \bigcap_{j = 0}^T \{ x_k = \tilde{s}_k , \bu_k = \tilde{\ba}_k \} \middle| x_0 = s, u^i_0 = a^i \right]  \\
= &\textbf{1} \{ \tilde{s}_0 = s, \tilde{a}^i_0 = a^i \} \cdot \prod_{ j \not= i } \pi^j ( \tilde{a}^j_0 | \tilde{s}_0 )   \times \prod_{k = 0}^{T-1} P \left( \tilde{s}_{k+1} | \tilde{s}_k , \tilde{\ba}_k \right) \cdot \prod_{k = 1}^T \bpi \left( \tilde{\ba}_k \middle| \tilde{s}_k \right),
\end{align*}
where $\bpi \left( \tilde{\ba}_k \middle| \tilde{s}_k \right) = \prod_{j \in \NN} \pi^j ( \tilde{a}^j_k | \tilde{s}_k )$. Since the conditional probability under $\P^{\bpi}$ of a given history of finite length is a finite product involving the components of $\bpi$, continuity follows.
\end{proof}

\subsection*{Proof of Lemma \ref{lemma:continuous-cost}} Fix player $i \in \NN, s \in \xx$, $\epsilon > 0$, and choose $T \in \nn$ large enough that 
\[  \frac{ ( \beta^i)^T }{ 1 - \beta^i } \cdot \max \left\{ \left| c^i ( \tilde{s}, \ba ) \right| : (\tilde{s}, \ba) \in \xx \times \bU \right\} < \frac{\epsilon}{4}  .
\]

\noindent Then, since 
\[
J^i ( \bpi ,  s ) = E^{\bpi} \left[ \sum_{t = 0}^T  (\beta^i)^t c^i ( x_t, \bu_t ) \middle| x_0 = s \right] + E^{\bpi} \left[ \sum_{t > T}  (\beta^i)^t c^i ( x_t, \bu_t ) \middle| x_0 = s \right]
\] 
for any $\bpi \in \bGamma_{S}$, we have that, for any $\bpi, \tilde{\bpi} \in \bGamma_{S}$,
\begin{align*}
&\left| J^i ( \bpi ,s ) - J^i ( \tilde{\bpi} , s ) \right| \leq  \\
&\left| E^{\bpi} \left[ \sum_{t = 0}^T  (\beta^i)^t c^i ( x_t, \bu_t ) \middle| x_0 = s \right] - E^{\tilde{\bpi} } \left[ \sum_{t = 0}^T  (\beta^i)^t c^i ( x_t, \bu_t ) \middle| x_0 = s \right] \right| + \frac{\epsilon}{2}.
\end{align*}

\noindent Define a function $C^i : \left( \xx \times \bU  \right)^{T+1} \to \rr$ by 
\[
C^i ( \tilde{s}_0, \tilde{\ba}_0, \cdots, \tilde{s}_T, \tilde{\ba}_T )   := \sum_{t = 0}^T ( \beta^i)^t c^i ( \tilde{s}_t, \tilde{\ba}_t ), \quad \forall ( \tilde{s}_k ,\tilde{\ba}_k )_{k=0}^T \in \left( \xx \times \bU \right)^{T+1}.
\]
We thus write
\begin{align*}
E^{\bpi} \left[ \sum_{t = 0}^T ( \beta^i)^t c^i( x_t, \bu_t ) \middle| x_0 = s \right] 
= \sum_{ \omega \in \left( \xx \times \bU \right)^{T+1} } C^i ( \omega )  \P^{\bpi} \left( 	(  x_k, \bu_k )_{k = 0}^T = \omega \middle| x_0 = s \right),
\end{align*}
for any $\bpi \in \bGamma_{S}$. From the latter expression and Lemma~\ref{lemma:continuous-probabilities}, we see that the mapping $\bpi \mapsto E^{\bpi} \left( \sum_{t = 0}^T ( \beta^i)^t c^i( x_t, \bu_t ) \middle| x_0 = s  \right) $ is continuous on $\bGamma_{S}$. 

Thus, taking $\bpi, \tilde{\bpi} \in \bGamma_{S}$ sufficiently close, under the metric $\textbf{d}$, we have that
\[
\left| E^{\bpi} \left[ \sum_{t = 0}^T  (\beta^i)^t c^i ( x_t, \bu_t ) \middle| x_0 = s \right] - E^{\tilde{\bpi} } \left[ \sum_{t = 0}^T  (\beta^i)^t c^i ( x_t, \bu_t ) \middle| x_0 = s \right] \right| < \frac{\epsilon}{2}.
\]
From our choice of $T$, we then also get that $\left| J^i ( \bpi , s)  - J^i ( \tilde{\bpi} , s ) \right| < \epsilon $, proving the lemma.

\subsection*{Proof of Lemma  \ref{lemma:continuous-Q-factors}} Fix player $i \in \NN$ and $(s,a^i ) \in \xx \times \uu^i$. We will show that the mapping $\bpi^{-i} \mapsto Q^{*i}_{\bpi^{-i}} ( s, a^i)$ is continuous on $\bGamma^{-i}_{S}$. Recall that for any $\bpi^{-i} \in \bGamma^{-i}_{S}$, the Q-factors are given by 
\[
Q^{*i}_{\bpi^{-i}} ( s, a^i ) := E^{ ( \pi^{*i} , \bpi^{-i} ) } \left[ \sum_{t = 0}^{\infty}  ( \beta^i )^t c^i ( x_t, \bu_t ) \middle| x_0 = s , u^i_0 = a^i \right],
\]
where $\pi^{*i} \in \BR^i_0 ( \bpi^{-i} )$ is any best-response to $\bpi^{-i}$. Since $\bpi^{-i} \in \bGamma^{-i}_{S}$ is stationary, player $i$ faces an MDP with controlled state process $\{ x_t \}_{t \geq 0}$. Therefore, some stationary and deterministic best-response policy exists. We denote the set of stationary and deterministic policies for player $i$ by $\DS^i \subset \Gamma^i_{S}$:
\[
\DS^i := \left\{ \pi^i \in \Gamma^i_{s} \middle| \forall \tilde{s} \in \xx, \exists \tilde{a}^i \in \uu^i : \pi^i ( \tilde{a}^i | \tilde{s} ) = 1 \right\}.
\]

Note that the set $\DS^i$ is finite and can be identified with the finite set $\{ f^i : \xx \to \uu^i \}$. Let $\pi^{*i} \in \BR^i_0 ( \bpi^{-i} ) \cap \DS^i$. We then have that, for any $\tilde{s} \in \xx$,
\begin{equation} \label{eq:cost-is-min-over-deterministic-stationary}
J^i ( \pi^{*i} , \bpi^{-i}  , \tilde{s} ) = \inf_{ \tilde{\pi}^i \in \Gamma^i } J^i  ( \tilde{\pi}^i , \bpi^{-i} , \tilde{s} ) = \min_{ \tilde{\pi}^i \in \DS^i } J^i  ( \tilde{\pi}^i , \bpi^{-i}  ,  \tilde{s} ).
\end{equation}

For each $\tilde{s} \in \xx$, the mapping $\bpi^{-i} \mapsto \min_{ \tilde{\pi}^i \in \DS^i } J^i ( \tilde{\pi}^i, \bpi^{-i}  , \tilde{s})$ is continuous on $\bGamma^{-i}_{S}$ by Lemma~\ref{lemma:continuous-cost}, as it is the pointwise minimum of finitely many continuous functions. 

Letting $\bpi^{*} = ( \pi^{*i}, \bpi^{-i} )$, we write $Q^{*i}_{\bpi^{-i}} ( s ,a^i )$ as 
\begin{align*}
&Q^{*i}_{\bpi^{-i}} ( s ,a^i )  	= E^{  \bpi^{*} } \left[ \sum_{ t = 0}^{\infty}  (\beta^i )^t c^i ( x_t, \bu_t ) \middle| x_0 = s, u^i_0 = a^i 		\right] \\
					= &E^{  \bpi^{*} } \left[ c^i ( x_0 , \bu_0 ) \middle| x_0 = s , u^i_0 = a^i  \right] + \beta^i \cdot E^{  \bpi^{*} } \left[   J^i ( \bpi^{*} , x_1 )  \middle| x_0 = s, u^i_0 = a^i \right],
\end{align*}
where the latter term is obtained by the tower property (conditioning on $x_0, u^i_0$, and $x_1$) and using the fact that $\bpi^{*}$ is stationary to simplify the resulting conditional expectation. Using \eqref{eq:cost-is-min-over-deterministic-stationary}, we then have 
\begin{align*}
& Q^{*i}_{\bpi^{-i}} ( s ,a^i )  	 \\
=& E^{  \bpi^{*} } \left[ c^i ( x_0 , \bu_0 ) \middle| x_0 = s , u^i_0 = a^i  \right] + \beta^i \cdot E^{  \bpi^{*} } \left[   \min_{ \tilde{\pi}^i \in \DS^i } J^i  ( \tilde{\pi}^i,  \bpi^{-i}  , x_1 )  \middle| x_0 = s, u^i_0 = a^i \right] .
\end{align*}

\noindent We observe that the first term does not depend on $\pi^{*i}$ and can be re-written as
\[
E^{  \bpi^{*} } \left[ c^i ( x_0 , \bu_0 ) \middle| x_0 = s , u^i_0 = a^i  \right]  = \sum_{ \ba^{-i} \in \bU^{-i}} \bpi^{-i} ( \ba^{-i} | s )c^i ( s, a^i, \ba^{-i} ),
\]
where $\bpi^{-i} ( \ba^{-i} | s ) = \prod_{j \not= i } \pi^j ( a^j | s )$. Note, further, that the mapping 
\[
\bpi^{-i} \mapsto \sum_{ \ba^{-i} \in \bU^{-i}} \bpi^{-i} ( \ba^{-i} | s ) c^i ( s, a^i ,\ba^{-i} )
\] 
is continuous on $\bGamma^{-i}_{S}$. Indeed, the second term does not depend on $\pi^{*i}$ either: using iterated expectations and conditioning additionally on $\bu^{-i}_0$, we write  
\begin{align*}
&E^{ \bpi^{*}  } \left[ \min_{ \tilde{\pi}^i \in \DS^i } J^i ( \tilde{\pi}^i, \bpi^{-i} , x_1 ) \middle| x_0 = s, u^i_0 = a^i \right] \\
= &\sum_{ \ba^{-i} \in \bU^{-i}} \bpi^{-i} ( \ba^{-i} | s ) \left( \sum_{ s' \in \xx } \min_{ \tilde{\pi}^i \in \DS^i } J^i (\tilde{\pi}^i, \bpi^{-i} , s' ) \cdot P ( s' | s , a^i, \ba^{-i} )  \right).
\end{align*}

\noindent From this, one sees that the mapping 
\[
\bpi^{-i} \mapsto \beta^i E^{ \bpi^{*}  } \left[ \min_{ \tilde{\pi}^i \in \DS^i } J^i  ( \tilde{\pi}^i, \bpi^{-i} , x_1 ) \middle| x_0 = s, u^i_0 = a^i \right]
\]
is also continuous on $\bGamma^{-i}_{S}$. Thus, $Q^{*i}_{\bpi^{-i}} ( s, a^i )$ is the sum of two functions that are continuous on $\bGamma^{-i}_{S}$ and so is itself continuous on $\bGamma^{-i}_{S}$, completing the proof.

\subsection*{Proofs of Lemmas \ref{lemma:continuous-min-action} and \ref{lemma:continuous-max-state}}

Lemma \ref{lemma:continuous-min-action} follows from Lemma~\ref{lemma:continuous-Q-factors}, as the function under consideration is the pointwise minimum of finitely many continuous functions.

Lemma~\ref{lemma:continuous-max-state} follows from Lemmas~\ref{lemma:continuous-cost} and \ref{lemma:continuous-min-action}, as the function under consideration is the pointwise maximum of finitely many continuous functions.

\newpage

\bibliographystyle{siamplain}
\bibliography{Bora-satisficing-feb2023}

\begin{thebibliography}{10}

\bibitem{agarwal2021theory}
{\sc A.~Agarwal, S.~M. Kakade, J.~D. Lee, and G.~Mahajan}, {\em On the theory
  of policy gradient methods: Optimality, approximation, and distribution
  shift}, Journal of Machine Learning Research, 22 (2021), pp.~1--76.

\bibitem{AY2017}
{\sc G.~Arslan and S.~Y\"{u}ksel}, {\em Decentralized {Q}-learning for
  stochastic teams and games}, IEEE Transactions on Automatic Control, 62
  (2017), pp.~1545--1558.

\bibitem{azar2011speedy}
{\sc M.~G. Azar, R.~Munos, M.~Ghavamzadeh, and H.~Kappen}, {\em Speedy
  {Q}-learning}, in Advances in Neural Information Processing Systems, 2011.

\bibitem{bai2020near}
{\sc Y.~Bai, C.~Jin, and T.~Yu}, {\em Near-optimal reinforcement learning with
  self-play}, Advances in Neural Information Processing Systems, 33 (2020),
  pp.~2159--2170.

\bibitem{Bowling}
{\sc M.~Bowling and M.~Veloso}, {\em Multiagent learning using a variable
  learning rate}, Artificial Intelligence, 136 (2002), pp.~215--250.

\bibitem{Brown1951iterative}
{\sc G.~W. Brown}, {\em Iterative solution of games by fictitious play},
  Activity Analysis of Production and Allocation, 13 (1951), pp.~374--376.

\bibitem{brown2018superhuman}
{\sc N.~Brown and T.~Sandholm}, {\em Superhuman {AI} for heads-up no-limit
  poker: Libratus beats top professionals}, Science, 359 (2018), pp.~418--424.

\bibitem{candogan2013near}
{\sc O.~Candogan, A.~Ozdaglar, and P.~A. Parrilo}, {\em Near-potential games:
  Geometry and dynamics}, ACM Transactions on Economics and Computation (TEAC),
  1 (2013), pp.~1--32.

\bibitem{carmona2016mean}
{\sc R.~Carmona, F.~Delarue, and D.~Lacker}, {\em Mean field games with common
  noise}, The Annals of Probability, 44 (2016), pp.~3740--3803.

\bibitem{cassidy1972solution}
{\sc R.~Cassidy, C.~Field, and M.~Kirby}, {\em Solution of a satisficing model
  for random payoff games}, Management Science, 19 (1972), pp.~266--271.

\bibitem{charnes1963deterministic}
{\sc A.~Charnes and W.~W. Cooper}, {\em Deterministic equivalents for
  optimizing and satisficing under chance constraints}, Operations Research, 11
  (1963), pp.~18--39.

\bibitem{Chasparis2013aspiration}
{\sc G.~C. Chasparis, A.~Arapostathis, and J.~S. Shamma}, {\em Aspiration
  learning in coordination games}, {SIAM} Journal on Control and Optimization,
  51 (2013), pp.~465--490.

\bibitem{chevalier2015micro}
{\sc G.~Chevalier, J.~Le~Ny, and R.~Malham{\'e}}, {\em A micro-macro traffic
  model based on mean-field games}, in 2015 American Control Conference (ACC),
  IEEE, 2015, pp.~1983--1988.

\bibitem{chien2011convergence}
{\sc S.~Chien and A.~Sinclair}, {\em Convergence to approximate {N}ash
  equilibria in congestion games}, Games and Economic Behavior, 71 (2011),
  pp.~315--327.

\bibitem{Claus1998}
{\sc C.~Claus and C.~Boutilier}, {\em The dynamics of reinforcement learning in
  cooperative multiagent systems}, in Proceedings of the Tenth Innovative
  Applications of Artificial Intelligence Conference, Madison, Wisconsin, 1998,
  pp.~746--752.

\bibitem{daskalakis2020independent}
{\sc C.~Daskalakis, D.~J. Foster, and N.~Golowich}, {\em Independent policy
  gradient methods for competitive reinforcement learning}, Advances in Neural
  Information Processing Systems, 33 (2020), pp.~5527--5540.

\bibitem{daskalakis2021independent}
{\sc C.~Daskalakis, D.~J. Foster, and N.~Golowich}, {\em Independent policy
  gradient methods for competitive reinforcement learning}, arXiv preprint
  arXiv:2101.04233,  (2021).

\bibitem{daskalakis2022complexity}
{\sc C.~Daskalakis, N.~Golowich, and K.~Zhang}, {\em The complexity of {M}arkov
  equilibrium in stochastic games}, arXiv preprint arXiv:2204.03991,  (2022).

\bibitem{devraj2017zap}
{\sc A.~M. Devraj and S.~Meyn}, {\em Zap {Q}-learning}, Advances in Neural
  Information Processing Systems, 30 (2017).

\bibitem{eksin2017distributed}
{\sc C.~Eksin and A.~Ribeiro}, {\em Distributed fictitious play for multiagent
  systems in uncertain environments}, IEEE Transactions on Automatic Control,
  63 (2017), pp.~1177--1184.

\bibitem{even2003learning}
{\sc E.~Even-Dar and Y.~Mansour}, {\em Learning rates for {Q}-learning},
  Journal of Machine Learning Research, 5 (2003), pp.~1--25.

\bibitem{fink1964equilibrium}
{\sc A.~M. Fink}, {\em Equilibrium in a stochastic $n$-person game}, Journal of
  Science of the Hiroshima University, Series AI (Mathematics), 28 (1964),
  pp.~89--93.

\bibitem{fischer2017connection}
{\sc M.~Fischer}, {\em On the connection between symmetric $n$-player games and
  mean field games}, The Annals of Applied Probability, 27 (2017),
  pp.~757--810.

\bibitem{foster2006regret}
{\sc D.~Foster and H.~P. Young}, {\em Regret testing: Learning to play {N}ash
  equilibrium without knowing you have an opponent}, Theoretical Economics, 1
  (2006), pp.~341--367.

\bibitem{gaunersdorfer1995fictitious}
{\sc A.~Gaunersdorfer and J.~Hofbauer}, {\em Fictitious play, {S}hapley
  polygons, and the replicator equation}, Games and Economic Behavior, 11
  (1995), pp.~279--303.

\bibitem{germano2007global}
{\sc F.~Germano and G.~Lugosi}, {\em Global {N}ash convergence of {F}oster and
  {Y}oung's regret testing}, Games and Economic Behavior, 60 (2007),
  pp.~135--154.

\bibitem{hart2003uncoupled}
{\sc S.~Hart and A.~Mas-Colell}, {\em Uncoupled dynamics do not lead to {N}ash
  equilibrium}, American Economic Review, 93 (2003), pp.~1830--1836.

\bibitem{hart2006stochastic}
{\sc S.~Hart and A.~Mas-Colell}, {\em Stochastic uncoupled dynamics and {N}ash
  equilibrium}, Games and Economic Behavior, 57 (2006), pp.~286--303.

\bibitem{hernandez2017survey}
{\sc P.~Hernandez-Leal, M.~Kaisers, T.~Baarslag, and E.~M. de~Cote}, {\em A
  survey of learning in multiagent environments: Dealing with
  non-stationarity}, arXiv preprint arXiv:1707.09183,  (2017).

\bibitem{hofbauer2003evolutionary}
{\sc J.~Hofbauer and K.~Sigmund}, {\em Evolutionary game dynamics}, Bulletin of
  the American Mathematical Society, 40 (2003), pp.~479--519.

\bibitem{Hu2003}
{\sc J.~Hu and M.~P. Wellman}, {\em Nash {Q}-learning for general-sum
  stochastic games}, Journal of Machine Learning Research, 4 (2003),
  pp.~1039--1069.

\bibitem{huang2007large}
{\sc M.~Huang, P.~E. Caines, and R.~P. Malham{\'e}}, {\em Large-population
  cost-coupled {LQG} problems with nonuniform agents: individual-mass behavior
  and decentralized $\epsilon$-{N}ash equilibria}, IEEE Transactions on
  Automatic Control, 52 (2007), pp.~1560--1571.

\bibitem{huang2006large}
{\sc M.~Huang, R.~P. Malham{\'e}, and P.~E. Caines}, {\em Large population
  stochastic dynamic games: closed-loop {M}c{K}ean-{V}lasov systems and the
  {N}ash certainty equivalence principle}, Communications in Information \&
  Systems, 6 (2006), pp.~221--252.

\bibitem{lasry2007mean}
{\sc J.-M. Lasry and P.-L. Lions}, {\em Mean field games}, Japanese Journal of
  Mathematics, 2 (2007), pp.~229--260.

\bibitem{lauer2000}
{\sc M.~Lauer and M.~Riedmiller}, {\em An algorithm for distributed
  reinforcement learning in cooperative multi-agent systems}, in In Proceedings
  of the Seventeenth International Conference on Machine Learning, Citeseer,
  2000.

\bibitem{leslie2005individual}
{\sc D.~Leslie and E.~Collins}, {\em Individual {Q}-learning in normal form
  games}, SIAM Journal on Control and Optimization, 44 (2005), pp.~495--514.

\bibitem{leslie2020best}
{\sc D.~S. Leslie, S.~Perkins, and Z.~Xu}, {\em Best-response dynamics in
  zero-sum stochastic games}, Journal of Economic Theory, 189 (2020),
  p.~105095.

\bibitem{levy2013discounted}
{\sc Y.~Levy}, {\em Discounted stochastic games with no stationary {N}ash
  equilibrium: two examples}, Econometrica, 81 (2013), pp.~1973--2007.

\bibitem{Littman1994}
{\sc M.~L. Littman}, {\em Markov games as a framework for multi-agent
  reinforcement learning}, in Machine Learning Proceedings 1994, Elsevier,
  1994, pp.~157--163.

\bibitem{Littman2001ffq}
{\sc M.~L. Littman}, {\em Friend-or-foe {Q}-learning in general-sum games}, in
  ICML, vol.~1, 2001, pp.~322--328.

\bibitem{liu2021sharp}
{\sc Q.~Liu, T.~Yu, Y.~Bai, and C.~Jin}, {\em A sharp analysis of model-based
  reinforcement learning with self-play}, in International Conference on
  Machine Learning, PMLR, 2021, pp.~7001--7010.

\bibitem{marden2012revisiting}
{\sc J.~R. Marden and J.~S. Shamma}, {\em Revisiting log-linear learning:
  Asynchrony, completeness and payoff-based implementation}, Games and Economic
  Behavior, 75 (2012), pp.~788--808.

\bibitem{Marden2009payoff}
{\sc J.~R. Marden, H.~P. Young, G.~Arslan, and J.~S. Shamma}, {\em Payoff-based
  dynamics for multiplayer weakly acyclic games}, SIAM Journal on Control and
  Optimization, 48 (2009), pp.~373--396.

\bibitem{marden2014achieving}
{\sc J.~R. Marden, H.~P. Young, and L.~Y. Pao}, {\em Achieving {P}areto
  optimality through distributed learning}, SIAM Journal on Control and
  Optimization, 52 (2014), pp.~2753--2770.

\bibitem{matignon2007hysteretic}
{\sc L.~Matignon, G.~J. Laurent, and N.~Le~Fort-Piat}, {\em Hysteretic
  {Q}-learning: an algorithm for decentralized reinforcement learning in
  cooperative multi-agent teams}, in 2007 IEEE/RSJ International Conference on
  Intelligent Robots and Systems, IEEE, 2007, pp.~64--69.

\bibitem{matignon2009coordination}
{\sc L.~Matignon, G.~J. Laurent, and N.~Le~Fort-Piat}, {\em Coordination of
  independent learners in cooperative {M}arkov games.}, HAL preprint
  hal-00370889,  (2009).

\bibitem{matignon2012survey}
{\sc L.~Matignon, G.~J. Laurent, and N.~Le~Fort-Piat}, {\em Independent
  reinforcement learners in cooperative {M}arkov games: a survey regarding
  coordination problems.}, Knowledge Engineering Review, 27 (2012), pp.~1--31.

\bibitem{matsui1992best}
{\sc A.~Matsui}, {\em Best response dynamics and socially stable strategies},
  Journal of Economic Theory, 57 (1992), pp.~343--362.

\bibitem{mnih2015human}
{\sc V.~Mnih, K.~Kavukcuoglu, D.~Silver, A.~A. Rusu, J.~Veness, M.~G.
  Bellemare, A.~Graves, M.~Riedmiller, A.~K. Fidjeland, G.~Ostrovski, et~al.},
  {\em Human-level control through deep reinforcement learning}, {N}ature, 518
  (2015), pp.~529--533.

\bibitem{ornik2018deception}
{\sc M.~Ornik and U.~Topcu}, {\em Deception in optimal control}, in 2018 56th
  Annual Allerton Conference on Communication, Control, and Computing
  (Allerton), IEEE, 2018, pp.~821--828.

\bibitem{ozdaglar2021independent}
{\sc A.~Ozdaglar, M.~O. Sayin, and K.~Zhang}, {\em Independent learning in
  stochastic games}, arXiv preprint arXiv:2111.11743,  (2021).

\bibitem{radner1975satisficing}
{\sc R.~Radner}, {\em Satisficing}, in Optimization Techniques IFIP Technical
  Conference, Springer, 1975, pp.~252--263.

\bibitem{Robinson1951iterative}
{\sc J.~Robinson}, {\em An iterative method of solving a game}, Annals of
  Mathematics, {\bf 54} (1951), pp.~296--301.

\bibitem{roughgarden2010algorithmic}
{\sc T.~Roughgarden}, {\em Algorithmic game theory}, Communications of the ACM,
  53 (2010), pp.~78--86.

\bibitem{sanjari2020optimality}
{\sc S.~Sanjari, N.~Saldi, and S.~Y{\"u}ksel}, {\em Optimality of independently
  randomized symmetric policies for exchangeable stochastic teams with
  infinitely many decision makers}, arXiv preprint arXiv:2008.11570,  (2020).

\bibitem{sanjari2018optimal}
{\sc S.~Sanjari and S.~Y{\"u}ksel}, {\em Optimal solutions to infinite-player
  stochastic teams and mean-field teams}, IEEE Transactions on Automatic
  Control, 66 (2020), pp.~1071--1086.

\bibitem{sanjari2019optimal}
{\sc S.~Sanjari and S.~Y{\"u}ksel}, {\em Optimal policies for convex symmetric
  stochastic dynamic teams and their mean-field limit}, SIAM Journal on Control
  and Optimization, 59 (2021), pp.~777--804.

\bibitem{sayin2021decentralized}
{\sc M.~Sayin, K.~Zhang, D.~Leslie, T.~Basar, and A.~Ozdaglar}, {\em
  Decentralized {Q}-learning in zero-sum {M}arkov games}, Advances in Neural
  Information Processing Systems, 34 (2021), pp.~18320--18334.

\bibitem{sayin2022fictitious}
{\sc M.~O. Sayin, F.~Parise, and A.~Ozdaglar}, {\em Fictitious play in zero-sum
  stochastic games}, SIAM Journal on Control and Optimization, 60 (2022),
  pp.~2095--2114.

\bibitem{schulman2015trust}
{\sc J.~Schulman, S.~Levine, P.~Abbeel, M.~Jordan, and P.~Moritz}, {\em Trust
  region policy optimization}, in International Conference on Machine Learning,
  PMLR, 2015, pp.~1889--1897.

\bibitem{schulman2017proximal}
{\sc J.~Schulman, F.~Wolski, P.~Dhariwal, A.~Radford, and O.~Klimov}, {\em
  Proximal policy optimization algorithms}, arXiv preprint arXiv:1707.06347,
  (2017).

\bibitem{sen94}
{\sc S.~Sen, M.~Sekaran, and J.~Hale}, {\em Learning to coordinate without
  sharing information}, in Proceedings of the 12th National Conference on
  Artificial Intelligence, 1994, pp.~426--431.

\bibitem{silver2016mastering}
{\sc D.~Silver, A.~Huang, C.~J. Maddison, A.~Guez, L.~Sifre, G.~Van
  Den~Driessche, J.~Schrittwieser, I.~Antonoglou, V.~Panneershelvam,
  M.~Lanctot, et~al.}, {\em Mastering the game of {G}o with deep neural
  networks and tree search}, Nature, 529 (2016), pp.~484--489.

\bibitem{silver2017mastering}
{\sc D.~Silver, J.~Schrittwieser, K.~Simonyan, I.~Antonoglou, A.~Huang,
  A.~Guez, T.~Hubert, L.~Baker, M.~Lai, A.~Bolton, et~al.}, {\em Mastering the
  game of {G}o without human knowledge}, Nature, 550 (2017), pp.~354--359.

\bibitem{simon1956rational}
{\sc H.~A. Simon}, {\em Rational choice and the structure of the environment.},
  Psychological review, 63 (1956), p.~129.

\bibitem{sutton2018reinforcement}
{\sc R.~S. Sutton and A.~G. Barto}, {\em Reinforcement learning: An
  introduction}, MIT press, 2018.

\bibitem{swenson2018distributed}
{\sc B.~Swenson, C.~Eksin, S.~Kar, and A.~Ribeiro}, {\em Distributed inertial
  best-response dynamics}, IEEE Transactions on Automatic Control, 63 (2018),
  pp.~4294--4300.

\bibitem{szepesvari1997asymptotic}
{\sc C.~Szepesv{\'a}ri}, {\em The asymptotic convergence-rate of {Q}-learning},
  Advances in neural information processing systems, 10 (1997).

\bibitem{tan1993multi}
{\sc M.~Tan}, {\em Multi-agent reinforcement learning: Independent vs.
  cooperative agents}, in Proceedings of the Tenth International Conference on
  Machine Learning, 1993, pp.~330--337.

\bibitem{tsitsiklis1994asynchronous}
{\sc J.~N. Tsitsiklis}, {\em Asynchronous stochastic approximation and
  {Q}-learning}, Machine Learning, 16 (1994), pp.~185--202.

\bibitem{Watkins89}
{\sc C.~Watkins}, {\em Learning from Delayed Rewards}, PhD thesis, Cambridge
  University, 1989.

\bibitem{WatkinsDayan92}
{\sc C.~Watkins and P.~Dayan}, {\em Q-{L}earning}, Machine Learning, {{\bf{
  8}}} (1992), pp.~279--292.

\bibitem{wei2017online}
{\sc C.-Y. Wei, Y.-T. Hong, and C.-J. Lu}, {\em Online reinforcement learning
  in stochastic games}, Advances in Neural Information Processing Systems, 30
  (2017).

\bibitem{wei2016lenient}
{\sc E.~Wei and S.~Luke}, {\em Lenient learning in independent-learner
  stochastic cooperative games}, The Journal of Machine Learning Research, 17
  (2016), pp.~2914--2955.

\bibitem{YAY-TAC}
{\sc B.~Yongacoglu, G.~Arslan, and S.~Y{\"u}ksel}, {\em Decentralized learning
  for optimality in stochastic dynamic teams and games with local control and
  global state information}, IEEE Transactions on Automatic Control, 67 (2021),
  pp.~5230--5245.

\bibitem{yongacoglu2022independent}
{\sc B.~Yongacoglu, G.~Arslan, and S.~Y{\"u}ksel}, {\em Independent learning
  and subjectivity in mean-field games}, in 2022 IEEE 61st Conference on
  Decision and Control (CDC), IEEE, 2022, pp.~2845--2850.

\bibitem{zhang2021multi}
{\sc K.~Zhang, Z.~Yang, and T.~Ba{\c{s}}ar}, {\em Multi-agent reinforcement
  learning: A selective overview of theories and algorithms}, Handbook of
  Reinforcement Learning and Control,  (2021), pp.~321--384.

\end{thebibliography}

\end{document}